\documentclass{article}
\usepackage{amsmath}
\usepackage{amssymb}
\usepackage{amsfonts}
\usepackage{a4wide}
\usepackage{amsthm}
\usepackage{tikz}
\usepackage{float}

\newcounter{NN}
\newtheorem{proposition}[NN]{Proposition}
\newtheorem{theorem}[NN]{Theorem}
\newtheorem{corollary}[NN]{Corollary}

\newtheorem{example}[NN]{Example}

\def\arXiv{{arXiv:}}
\def\url{}
\begin{document}
\bibliographystyle{plain}
\title{Generalised Manin transformations and QRT maps}
\author{Peter H.~van der Kamp, David I. McLaren and G.R.W. Quispel}
\date{Department of Mathematics and Statistics\\
La Trobe University, Victoria 3086, Australia\\[5mm]
measure preservation,  Pascal’s theorem, symmetries.
}

\maketitle

\begin{abstract}
Manin transformations are maps of the plane that preserve a pencil of cubic curves. They are the composition of two involutions. Each involution is constructed in terms of an involution point that is required to be one of the base points of the pencil.
We generalise this construction to explicit birational maps of the plane that preserve quadratic resp. certain quartic pencils, and show that they are measure-preserving and hence integrable. In the quartic construction the two involution points are required to be base points of the pencil of multiplicity 2. On the other hand, for the quadratic pencils the involution points can be any two distinct points in the plane (except for base points). We employ Pascal's theorem to show that the maps that preserve a quadratic pencil admit infinitely many symmetries. The full 18-parameter QRT map is obtained as a special instance of the quartic case in a limit where the two involution points go to infinity. We show by construction that each generalised Manin transformation can be brought to QRT form by a fractional affine transformation. We also specify classes of generalised Manin transformations which admit a root.
\end{abstract}

\section{Introduction}
The (18-parameter) Quispel-Roberts-Thompson (QRT) map \cite{QRT1,QRT2} has become an archetypical integrable map of the plane.  It is measure preserving, preserves a pencil of biquadratic curves, and can be written as the composition of 2 involutions. Starting with works of Tsuda \cite{Tsu}, Jogia et al \cite{Jet} and Duistermaat's monograph \cite{Dui}, a thorough understanding of these maps was provided from an algebraic geometric viewpoint. For example, as shown by Tsuda, the QRT map can be described as an addition formula on a rational elliptic surface.

A Manin transformation \cite{Man,Dui} is also an integrable map of the plane, measure preserving and a composition of two involutions. However, it leaves invariant a pencil of cubic curves. In \cite{CMMOQ2,PSS,PS,KCMMOQ} it was shown that Manin transformations arise in Kahan discretizations of certain vector fields. Other integrable maps which preserve pencils of different degree type can also be found in the literature, e.g. pencils of biquartic curves \cite{JGTR,KJ,Het,Ket,VGR}, and pencils of bisextic curves \cite{CMMOQ,PPS}.

It is known \cite{Jet,Ves} that every birational map of infinite order which preserves a pencil of algebraic curves is birationally conjugate to a translation, either on a ruled rational surface or on an elliptic surface. Hence, the pencils of the above mentioned maps all have genus $<2$. As the genus $g$ of a curve with $n_m$ singular points of multiplicity $m$ is related to the degree $N$,
\begin{equation} \label{genus}
g=\frac{(N-1)(N-2)}2-\sum_{m} n_m \frac{m(m-1)}2,
\end{equation}
and the genus is invariant under birational transformations, it follows that for $N>3$ the invariant pencils of these maps have singular points.

One can ask for a given map which preserves a pencil of algebraic curves whether it is birationally conjugate to a QRT map, or, quoting \cite{VGR}
\begin{quote} \label{quote}
whether all integrable second order mappings with a rational invariant can be brought to a QRT form by a birational change of coordinates of the 2-plane.
\end{quote}
The authors of \cite{VGR} provide two examples of maps which both preserve a quartic pencil, but for only one of these they were able to construct a transformation to a QRT mapping. As an example, we show that that map is the root of a generalised Manin transformation, cf. section 7.

In this paper we provide a geometric construction of classes of mappings which preserve a pencil of curves of total degree $N=2,3,4$ (thus including all Manin transformations), and show these can all be brought into QRT form by a projective collineation. In the case $N=2$, our construction gives rise to the existence of uncountably many symmetries, through application of Pascal's theorem.

We also specify which subclasses of mappings are equivalent to a root of a QRT map, also known as a symmetric QRT mapping. These includes mappings which arise as the Kahan discretisation of physical systems such as the Suslov motion of a rigid body under the constraint that a certain component of the angular velocity vector vanishes \cite{susl}, and symmetric monopoles as described by reduced Nahm equations \cite{Nahm}, cf. \cite{KCMMOQ}.

\section{Generalised Manin transformations}
In a biquadratic pencil, a horizontal (or vertical) line intersects a generic curve in two points only. Hence one can define a horizontal (vertical) switch $\iota_1$ ($\iota_2$) as the involution which switches those two points. This geometric construction defines the QRT map,  $\tau=\iota_2\circ\iota_1$, cf. \cite[page viii]{Dui}.

A Manin transformation is also a composition of two involutions. They are defined for cubic curves \cite{Man}, see also \cite[Section 4.2]{Dui}. Given a base point $p$ of a cubic pencil, i.e. a point which lies on every curve in the pencil, the line through $p$ intersects each curve in only two other points. Hence one can define a $p$-switch $\iota_p$ as the involution which switches those two points. We call $p$ the involution point of  $\iota_p$. If $q$ is another base point, a Manin transformation is obtained by composition, $\tau_{p,q}=\iota_q\circ\iota_p$.

This geometric construction can be generalised to pencils of degree $N\geq2$, $P_{\alpha,\beta}(u,v)=0$, where
\begin{equation} \label{pen}
P_{\alpha,\beta}(u,v):=\alpha F_a(u,v)+\beta F_b(u,v),
\end{equation}
and $F_e(u,v)$ is a polynomial in two variables $u,v$ of fixed total degree $N$ which depends on parameters $e_1,e_2,\ldots$, and $F_a\neq F_b$. If the degrees of $F_a$ and $F_b$ are not equal then we take the degree of the pencil, $N$, to be the largest of the two degrees. For all $(u,v)$ there are $\alpha,\beta$ such that $P_{\alpha,\beta}(u,v)=0$, i.e. $\frac{\alpha}{\beta}=-\frac{F_b}{F_a}(u,v)$. For base points $(u,v)$ we have $P_{\alpha,\beta}(u,v)=0$ for all $\alpha,\beta$, and there are $N^2$ of them (considering $(u,v)$ to be projective coordinates in $\mathbb{P}^2$ and counting intersection multiplicities\footnote{The reader should be aware that by doing so QRT-maps have generically 10 base points including 2 singular point at $(0,\infty)$ and $(\infty,0)$ yielding an intersection total of $8(1\cdot 1)+2(2\cdot 2)=4\cdot4$, instead of the 8 base points in $\mathbb{P}^1\times\mathbb{P}^1$, cf. \cite[Lemma 3.1.1]{Dui}.}), namely the solutions of $F_a=F_b=0$.

For $N=2$ we are free to choose the involution points $p,q$ and there are no constraints on the pencil. For $N=3$ (the Manin case) there are no constraints on the pencil, the involution points are base points of the pencil. For $N=4$ we require the pencil to have two base points, $p$ and $q$, which are singular points (of multiplicity\footnote{A curve $C(u,v)=0$ has a singular point of multiplicity $m$ if $m\geq 1$ is the smallest number such that all $k$-th order partial derivatives with $k<m$ vanish at $(c,d)$ \cite{SF}. A singular point of multiplicity $m$ is also called a double point ($m=2$), a triple point ($m=3$), or an $m$-ple point.} 2), and which we choose to be the involution points. For $N>4$ the base points $p$ and $q$ are required to be singular points of multiplicity $N-2$. This ensures that any line through $p$ or $q$ intersects each curve of the degree $N$ pencil in only two points, and hence that the involutions $\iota_p$ and $\iota_q$ are well-defined. The construction here is reminiscent of the construction in \cite{ST}, page 99, where the group of rational points becomes quite different if one uses a line through a singular point of higher multiplicity.

Other geometric constructions of birational involutions have been found. In \cite{PSWZ}, involutions are defined using a pencil of curves of degree $M$, such that the intersection with a given pencil of curves of degree $N$ at the common base points is $MN-2$. In \cite{vdK}, the current construction is generalised by allowing involutions of the type $\iota_p$, where $p$ is not fixed but lies on a so called involution curve. We note that birational involutions of the plane have been classified by Bertini \cite{Bay,Ber}: every non-trivial birational involution of $\mathbb{P}^2$ is birationally conjugate to exactly one of the following: a de Jonquieres involution, a Geiser involution, or a Bertini involution. In the work of Moody \cite{Moo}, the Bertini involution has been described as a Manin transformation with an involution curve, although not in these terms. Using the results of \cite{CD,Dol}, maps preserving an elliptic fibration were classified in \cite{CT}: they i-$m$) preserve each fiber of a Halphen surface of index $m$, or ii-$m$) they do not preserve each fiber. We mention that Manin involutions are de Jonquieres, cf. \cite{PSWZ}. Furthermore, all transformations we construct are fiber preserving, of type i-$m$. The precise birational equivalence to mentioned mappings is beyond the scope of this paper.

We will now provide an explicit formula for the generalised Manin involution $\iota_p$ that preserves a pencil (\ref{pen}) of degree $N$, in terms of the polynomials $F_a$ and $F_b$ and their first and second order partial derivatives. The formula $(x,y)=(u+(c-u)z,v+(d-v)z)$ gives a parametrization of the line going through $(u,v)$, for $z=0$, and through $p=(c,d)$, for $z=1$. Below, in equation (\ref{sec}), we provide the value of $z$ such that $(x,y)$ and $(u,v)$ are on the same curve of the given pencil, i.e. such that $F_a(x,y)F_b(u,v)=F_a(u,v)F_b(x,y)$.
Denote $F_a(z):=F_a(u+(c-u)z,v+(d-v)z)$, and  $F_a^{(z)}:=\frac{\text{d}}{\text{d} z}F_a$. A Taylor expansion, about $z=0$, gives
\begin{equation} \label{TF}
F_a(z)=F_a(0)+F_a^{(z)}(0)z+\frac12F_a^{(z,z)}(0)z^2+\cdots+\frac1{N!}F_a^{(z,\overset{N}{\ldots},z)}(0)z^N,
\end{equation}
where
\[
F_a^{({z,\overset{n}{\ldots},z})}(0)=\sum_{i=0}^{n}
\binom{n}{i} F_a^{({u,\overset{i}{\ldots},u},{v,\overset{n-i}{\ldots},v})}(u,v)(c-u)^i(d-v)^{n-i}.
\]
For $N>2$ we have $
F_a(1)=F_a^{(z)}(1)=\cdots=F_a^{({z,\overset{N-3}{\ldots},z})}(1)=0$,
and similarly for $F_b$. These equations are used in appendix A to prove the explicit formula for the generalised Manin involution, given in the following theorem.
\begin{theorem} \label{T1}
Let $P_{\alpha,\beta}(u,v)=0$ be a pencil of degree $N\geq2$ and let $p$ be a point which for $N>2$ is a base point and has multiplicity $N-2$. Then the generalised Manin involution with involution point $p=(c,d)$ is given by
\begin{equation} \label{Man}
\iota_p:(u,v)\rightarrow (u,v)+z (c-u,d-v),
\end{equation}
where $z$ is given by
\begin{equation} \label{sec}
z=2\left(2(2-N)-\frac{F_a(0)F_b^{(z,z)}(0)-F_a^{(z,z)}(0)F_b(0)}{F_a(0)
F_b^{(z)}(0)-F_a^{(z)}(0)F_b(0)}\right)^{-1}.
\end{equation}
\end{theorem}
In Appendix B we derive a condition which enables one to verify that $\iota_p$ is anti measure preserving with density\footnote{Recall \cite[Section 2.2]{RQ} that a map $\phi$ is (anti) measure preserving with density $\rho$ if its Jacobian $J$ equals $(-)\rho/(\rho\circ\phi)$.}
$
L^{N-3}/F_a$,
where $L=0$ is a line through $p$, and we comment that the condition is satisfied for $N=2,3,4$.

The above construction provides an explicit formula for Manin involutions on pencils of any degree $N>1$, which (for $N>3$) admit a base point that is a singular point of multiplicity $N-2$.
From two distinct generalised Manin involutions (\ref{Man}), one can compose a generalised Manin transformation:
\begin{equation} \label{tpq}
\tau_{p,q}=\iota_q\circ\iota_p,
\end{equation}
which preserves a pencil of degree $N$. However, there are no generalised Manin transformations which preserve an irreducible pencil of degree $N>4$. According to (\ref{genus}) the genus of a curve of degree $N>3$ with two singular points of degree $N-2$ is $\frac{(N-2)(5-N)}{2}$, which is less than zero for $N>5$ and hence such curves are reducible. In Appendix C, we show by geometric means that curves with two singular points $p,q$ of degree $N-2$ are reducible for $N>4$ and that lines through $p$ and $q$ factor out. As a corollary, it follows that the generalised Manin transformation (\ref{tpq}), which preserves pencils of total degree $4$ is the most general.

Note that because a biquadratic polynomial is a special instance of a quartic polynomial with double points at $(\infty,0)$ and $(0,\infty)$, the full 18-parameter QRT map is obtained as a special case of the degree $N=4$ generalised Manin transformation.

\section{Transforming a generalised Manin transformation into QRT form} \label{QRTform}
For $N=2,3,4$ every generalised Manin transformation (\ref{tpq}) can be brought into QRT form (which can be regarded as a normal form for generalised Manin transformations) by a projective collineation which transforms the line through the involution points to infinity. For a given map, if it preserves a pencil of degree 3 or 4, it is easy to find the transformation: for $N=3$ the involution points are base points of the pencil, and for $N=4$ they are singular base points. In any case, the involution points are included in the set of base points of the map and its inverse.

Consider the fractional affine transformation
\begin{equation} \label{flsw}
\psi:(u,v)\rightarrow (U,V)= \left(\frac{au+bv+c}{gu+hv+i},\frac{du+ev+f}{gu+hv+i}\right).
\end{equation}
Such a transformation maps lines to lines, which can be seen as follows. The coordinates $(u,v)$ can be taken as affine coordinates of a projective space and then $\psi$ (\ref{flsw}) is induced
by a linear transformation of the vector space it is derived from. Indeed, we can write $\psi=\kappa \phi \kappa^{-1}$ where $\phi$ is a linear map and $\kappa:(u,v,w)\rightarrow(u/w,v/w)$.
Since $\kappa(p+t(q-p))=\kappa(p)+s(\kappa(q)-\kappa(p))$, with $sp_3-tq_3=ts(p_3-q_3)$ the maps $\kappa$, $\kappa^{-1}$, and hence $\psi$ (\ref{flsw}), map lines to lines. Such a map is called a homography, or, a {\em projective collineation}. The fundamental theorem of projective geometry states that every map which sends lines to lines (in a projective space of dimension at least two)
is a projective collineation \cite[Thm 2.26]{Art}.

If $p=(c,d)$ and $q=(e,f)$ are points in the plane and \begin{equation}\label{L}
L(u,v)=(d-f)(u-e)-(c-e)(v-f),
\end{equation}
so that $L=0$ is the line through $p$ and $q$, then any projective collineation of the form,
\begin{equation} \label{toqrt}
(u,v)\rightarrow \left(\frac{A(u-e)+B(v-f)}{L},\frac{C(u-c)+D(v-d)}{L}\right),
\end{equation}
where neither $(A,B)$ nor $(C,D)$ is perpendicular to $L$ (ensuring invertible), sends $p$ to $(\infty,0)$, and $q$ to $(0,\infty)$. Throughout this paper we will refer to the line $L=0$ through $p$ and $q$ as the Manin line for the generalised Manin transformation (\ref{tpq}).
Thus we have the following result.

\begin{theorem} Let $p=(c,d)$ and $q=(e,f)$ be the involution points for a pencil of curves $P_{\alpha,\beta}(u,v)=0$ of degree $2\leq N\leq 4$, so that if $N>2$ then $p,q$ are base points of multiplicity $N-2$. With $L=0$ being the Manin line, and for all $A,B,C,D$, the projective collineation (\ref{toqrt}) brings the generalised Manin transformation (\ref{tpq}) into QRT form.
\end{theorem}

In the remainder of this paper we consider the cases $N=2,3,4$ separately, and section \ref{roots} is devoted to the study of roots of generalised Manin transformations, which are equivalent to the so called symmetric QRT maps. As we will show there, the example considered in \cite[section 3]{VGR} turns out to be the root of a generalised Manin transformation.

\section{Quadratic pencils} \label{squad}
In this section  we consider the degree $N=2$ case. Taking two different involution points $p$ and $q$, the 16-parameter map $\tau=\iota_{q}\circ\iota_{p}$ is measure-preserving with density $1/\left(L(u,v)F_a(u,v)\right)$, where $L=0$ is the Manin line. Using Pascal's hexagrammum mysticum theorem, we show that for any $r$ on the Manin line the map $\iota_{r}$ is a reversing symmetry of $\tau$. This implies that the map $\tau$ has uncountably many symmetries.

Let
\begin{equation} \label{quadfa}
F_a(u,v):=a_1 + a_2 u + a_3 v + a_4 u^2 + a_5 u v + a_6 v^2
\end{equation}
be a polynomial of degree $N=2$ in variables $u,v$, that is $a_4,a_5$ and $a_6$ are not all zero. We have a pencil $P_{\alpha,\beta}(u,v)=0$ of conics (i.e. curves of genus zero). Any point $p=(c,d)$ can be taken as involution point. An involution is defined by
\begin{equation} \label{evol}
\iota_p(u,v)=(u,v)+z(c-u,d-v),
\end{equation}
with $z$ given by (\ref{sec}) (or alternatively by (\ref{fir})), where, explicitly,
\begin{equation}\label{expl}
\begin{split}
F_a(0)&=F_a(u,v),\\ F_a^{(z)}(0)&=F_a^{(u)}(u,v)(c-u)+F_a^{(v)}(u,v)(d-v),\\ F_a^{(z,z)}(0)&=F_a^{(u,u)}(u,v)(c-u)^2+2F_a^{(u,v)}(u,v)(c-u)(d-v)
+F_a^{(v,v)}(u,v)(d-v)^2,
\end{split}
\end{equation}
and
$F_a^{(u)}(u,v)=a_2+2a_4u+a_5v$, $F_a^{(v)}=a_3+a_5u+2a_6v$,
$F_a^{(u,u)}(u,v)=2a_4$,
$F_a^{(u,v)}(u,v)=a_5$,
$F_a^{(v,v)}(u,v)=2a_6$.

\begin{example} \label{exampquadratic}
Ten curves from the pencil $P_{\alpha,\beta}(u,v)=0$ with
\begin{equation}\label{Exampl2}
F_a(u,v)=u^2-uv+v^2+u-v-2 \quad \mbox{and} \quad  F_b(u,v)=uv,
\end{equation}
are plotted in Figure \ref{QP}.

\begin{figure}[htb]
\begin{center}
\includegraphics[width=10cm]{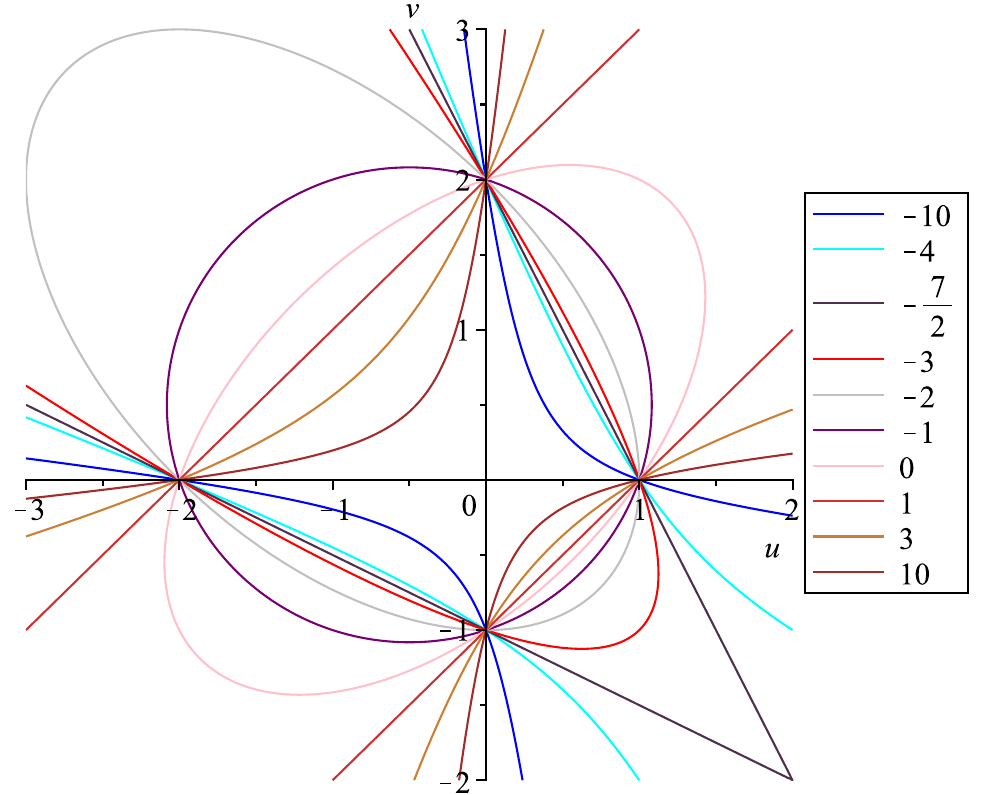}
\caption{\label{QP} Ten curves from the quadratic pencil defined by (\ref{pen}) and (\ref{Exampl2}), labeled by the value of $-\beta/\alpha$. The base points are $(1,0)$, $(0,-1)$, $(-2,0)$, $(2,0)$.}
\end{center}
\end{figure}

Taking $p=(2,-2)$ yields the involution
\begin{equation}\label{cd0}
\iota_{2,-2}(u,v)=-\frac{2}{u-v-2}(v,u),
\end{equation}
and taking $q=(-1,1)$ yields the involution
\begin{equation}\label{cd1}
\iota_{-1,1}(u,v)=\frac{\left(-v(2u+v+1),u(u+2v-1)\right)}{u^2+uv+v^2-1}.
\end{equation}

The Manin line is $u+v=0$. Introducing new coordinates
\[
\left(x,y\right)=\left(\frac{u+1}{u+v},\frac{v+2}{u+v}\right)
\]
the involution $\iota_p$ becomes
\[
\iota_1:(x,y)\rightarrow (y-x+\frac12,y)
\]
and the involution $\iota_q$ becomes
\[
\iota_2:(x,y)\rightarrow \left(x,\frac{x+2-xy}{x-y}\right).
\]
The ratio $F_a/F_b$ becomes
\[
\frac{{y}^{2}+14\,x\,(y-x)+7\,x-8\,y-2}{(2\,x-y+1)(2\,x-y-2)}.
\]
The QRT mapping $\tau=\iota_2 \circ \iota_1$ has matrices
\[
A^0=\begin{pmatrix}
0&0&-14\\
0&14&7\\
1&-8&-2
\end{pmatrix},\qquad A^1=
\begin{pmatrix}
0&0&4\\
0&-4&-2\\
1&1&-2
\end{pmatrix}.
\]
\end{example}

\noindent
In general, the involution (\ref{evol}) has the form
\[
\iota_p(u,v)=\frac{\big(N_1(u,v),N_2(u,v)\big)}{D(u,v)},
\]
where $N_i$ and $D$ are generically of degree $t=3$. If $t=3$ the point $p=(c,d)$ is a double point on $N_1=0$, $N_2=0$ and on $D=0$, and all points on the curve $C$ defined by $F_a(c,d)F_b(u,v)=F_b(c,d)F_a(u,v)$ are mapped to $(c,d)$.
When $p$ is a point on one line through two base points, the degree is lowered to $t=2$, and $p$ is a simple point on $N_1=0,N_2=0$, and on $D=0$.
An example is given by (\ref{cd1}). Here the map $\iota_p$ is an involution on the line that contains $p$, but the other line of the union $C$ is mapped to $p$.
When $p$ is the intersection of two straight lines through two base points, the degree is lowered to $t=1$ and $p$ is not on $N_1=0,N_2=0$ or on $D=0$. The involution is an involution on both lines, (\ref{cd0}) provides an example.
For base points $p$ the degree is $t=0$, i.e. we have $\iota_p=\text{id}$, the identity.

The involution $\iota_{c,d}$ (\ref{evol}) is anti measure-preserving with density
\begin{equation} \label{dens}
\frac{1}{\left(r(u-c)+s(v-d)\right)F_a(u,v)},
\end{equation}
where the first factor represents
any straight line through $(c,d)$. Taking the composition of two  involutions (\ref{evol}), we construct the map $\tau_{p,q}$ (\ref{tpq}).
The following holds.
\begin{proposition}
The map $\tau_{p,q}$ defined by (\ref{tpq}), which preserves each curve of the quadratic pencil $P_{\alpha,\beta}(u,v)=0$ with (\ref{quadfa}), is an integrable map of the plane. It is measure-preserving with density
$
(L(u,v)F_a(u,v))^{-1}$,
where $L(u,v)=0$ is given by (\ref{L}) and $L=0$ is the Manin line through the involution points $p=(c,d)$ and $q=(e,f)$.
\end{proposition}

Let us now define two special involutions,
\begin{equation} \label{sevol}
\iota_1=\lim_{c\rightarrow\infty} \iota_{c,0},\qquad
\iota_2=\lim_{f\rightarrow\infty} \iota_{0,f},
\end{equation}
the horizontal, respectively vertical, switch, cf. \cite[page viii]{Dui}.
Considering the involution (\ref{evol}), it is clear that $z$ is of the form $z=N/D$ where $N$ is linear in $c,d$, and $D$ quadratic. Hence, the involutions have the form $\iota_1(u,v)=(u+cz,v)$, and $\iota_2(u,v)=(u,v+fz)$. In the respective limits we find
\[
cz=-2\frac{F_a(u,v)F_b^{(u)}(u,v)-F_b(u,v)F_a^{(u)}(u,v)}{F_a(u,v)F_b^{(u,u)}(u,v)-F_b(u,v)F_a^{(u,u)}(u,v)},
\]
and
\[
fz=-2\frac{F_a(u,v)F_b^{(v)}(u,v)-F_b(u,v)F_a^{(v)}(u,v)}{F_a(u,v)F_b^{(v,v)}(u,v)-F_b(u,v)F_a^{(v,v)}(u,v)}.
\]
The map $\tau=\iota_2\circ\iota_1$ is a special case of the asymmetric QRT map \cite{QRT1,QRT2}, with matrices
\[
A^0=\begin{pmatrix}
0 & 0 & a_4 \\
0 & a_5 & a_2 \\
a_6 & a_3 & a_1
\end{pmatrix}\quad \text{ and } \quad
A^1=\begin{pmatrix}
0 & 0 & b_4 \\
0 & b_5 & b_2 \\
b_6 & b_3 & b_1
\end{pmatrix},
\]
cf. page 1 of Duistermaat's book \cite{Dui}. The involutions $\iota_1$ and $\iota_2$ (\ref{sevol}) are anti measure-preserving with densities $1/(F_a(u,v)(r_1v+r_2))$, $1/(F_a(u,v)(s_1u+s_2))$
respectively, for arbitrary $r_i,s_i$. This implies in particular that $\tau$ is measure-preserving with density
$1/F_a(u,v)$,
and $\iota_{c,d}\circ \iota_1$ is measure-preserving with density
$1/\left(\left(v-d\right)F_a(u,v)\right)$.

\subsection*{Symmetries}
The following theorem follows from Pascal's theorem \cite{Yze}, which is illustrated by Figure \ref{Pas}.

\begin{figure}[h]
\begin{center}
\includegraphics[width=9cm]{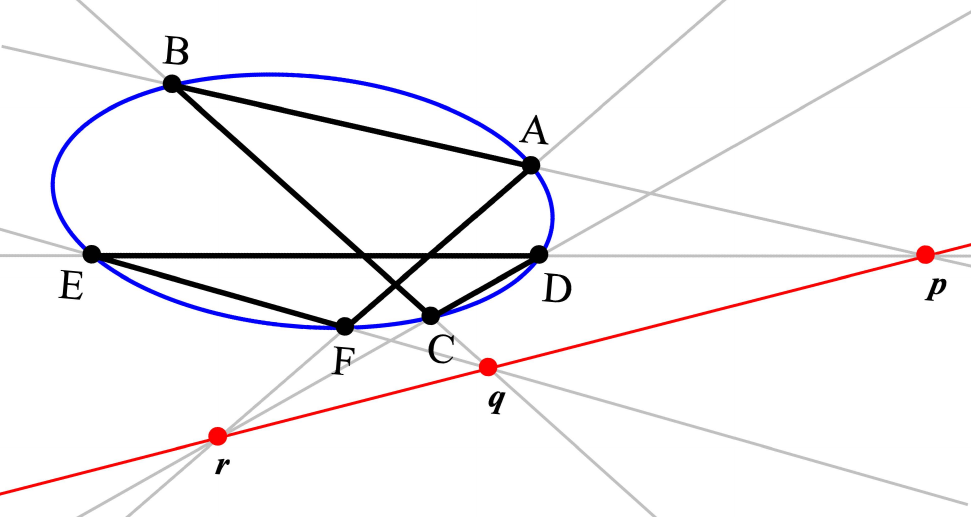}
\caption{\label{Pas} Lines through opposite sides of a hexagon on a conic meet in three points which lie on a straight line, called the Pascal line.}
\end{center}
\end{figure}

\begin{theorem} \label{theums}
A map $\tau_{p,q}$ defined by (\ref{tpq}), which preserves a quadratic pencil $P_{\alpha,\beta}(u,v)=0$, has uncountably many symmetries.
\end{theorem}

\begin{proof}
We first show that the map $\tau_{p,q}$ has uncountably many reversing symmetries, cf. \cite{RQ}. Let $r$ be on the line through $p$ and $q$, and let
\[
B=\iota_{p}(A),\quad
C=\iota_{q}(B),\quad
D=\iota_{r}(C),\quad
E=\iota_{p}(D),\quad
F=\iota_{q}(E),
\]
as in Figure \ref{Pas}. By construction $A,B,C,D,E,F$ lie on a conic. The lines $AB$ and $DE$ meet in $p$, the lines $BC$ and $EF$ meet in $q$. According to Pascal's theorem the lines $CD$ and $AF$ meet in a point $s$ on the Pascal line $pq$. But $r$ is on $CD$ and on $pq$, so we have $s=r$ and hence $A=\iota_{r}(F)$.
It follows that $\iota_{r}\circ\iota_{q}\circ\iota_{p}$ is an involution.
Thus, we have $\iota_{r}\tau_{p,q}=\tau_{p,q}^{-1}\iota_{r}$ showing that
$\iota_{r}$ is a reversing symmetry. Uncountably many symmetries are obtained by composition of reversing symmetries (and more reversing symmetries by composition of
symmetries and reversing symmetries).
\end{proof}

\begin{corollary} Theorem \ref{theums} implies that QRT maps which preserve a pencil of quadratic curves admit uncountably many reversing symmetries, namely all generalised Manin involutions with involution point at infinity.
\end{corollary}

\begin{example} For Example \ref{exampquadratic}, other involutions in the $(u,v)$-plane (\ref{evol}), whose involution point is on the line $u+v=0$ give rise to mappings that are reversing symmetries of the map $\tau$. Examples are $\iota_{0,0}$ which in QRT coordinates gives rise to
\[
(x,y)\rightarrow (x,y)-(y-\frac{1}{2})(1,2)
\]
and $\iota_{1,-1}$ which gives rise to
\begin{align*}
(x,y)&\rightarrow (x,y)-{\frac {4\,{x}^{2}y-10\,x{y}^{2}+4\,{y}^{3}-2\,{x}^{2}+6\,xy-{y}^{2}+
13\,x-10\,y-2}{2\,{x}^{2}-2\,xy-4\,{y}^{2}-x+5\,y+8}}
(2,1).
\end{align*}
\end{example}

\section{Cubic pencils} \label{scubi}
In this section, we consider the degree $N=3$ case, where the pencil comprises elliptic curves of genus 1. We parametrise the pencil in terms of the coordinates of two distinct base points $p$ and $q$, which we choose to be involution points. The 20-parameter map we obtain explicitly, $\tau=\iota_{q}\circ\iota_{p}$, is measure-preserving with density $1/F_a(u,v)$.

An irreducible plane curve of degree three with no singular points has genus one. Two such curves generically intersect in nine points. To find these intersection points, in general one needs to find the roots of a ninth order polynomial. However, we use the coordinates of two distinguished, and distinct, intersection points, $p=(c,d)$ and $q=(e,f)$, to parametrise the cubic curves. We require the cubics
\begin{equation}\label{cubs}
F_a(u,v):=a_1 + a_2 u + a_3 v + a_4 u^2 + a_5 u v + a_6 v^2 + a_7 u^3 + a_8 u^2 v + a_9 u v^2 +a_{10} v^3
\end{equation}
to vanish at $p$ and $q$. Assuming that $K:={c}^{3}{f}^{3}-{d}^{3}{e}^{3}$ does not vanish\footnote{One can also consider the case where $K=0$: if $c\neq e$ one can solve for $a_1$ and $a_2$, or when $d\neq f$ one can solve for $a_1$ and $a_3$.}, we can solve the constraints for the parameters $a_7$ and $a_{10}$. We find
$a_7 = G_a/K$, $a_{10} = H_a/K$ with
\begin{align*}
G_a&=
({d}^{3}-{f}^{3})a_1
+({d}^{3}e-c{f}^{3})a_2
+df({d}^{2}-{f}^{2})a_3
+({d}^{3}{e}^{2}-{c}^{2}{f}^{3})a_4 \\
&\ \ \ \
+df({d}^{2}e-c{f}^{2})a_5
+{d}^{2}{f}^{2}(d-f)a_6
+df({d}^{2}{e}^{2}-{c}^{2}{f}^{2})a_8
+{d}^{2}{f}^{2}(de-cf)a_9,
\\
H_a&=
({e}^{3}-{c}^{3})a_1
+ce(e^{2}-{c}^{2})a_2
+(d{e}^{3}-f{c}^{3})a_3
+c^2{e}^{2}(e-c)a_4\\
&\ \ \ \
+ce(d{e}^{2}-f{c}^{2})a_5
+(d^2{e}^{3}-f^2{c}^{3})a_6
+c^2e^{2}(de-fc)a_8
+ce(d^2e^{2}-f^2{c}^{2})a_9.
\end{align*}
We have chosen this parametrisation so we can easily set $d=e=0$ and take a limit where $c$ or $f$ goes to infinity, which yields $a_7=0$, $a_{10}=0$ respectively. If both limits are taken we are left with a biquadratic
\begin{equation}\label{CUFS}
F_a(u,v)={u}^{2}va_{{8}}+u{v}^{2}a_{{9}}+{u}^{2}a_{{4}}+uva_{{5}}+{v}^{2}a_{{6}
}+ua_{{2}}+va_{{3}}+a_{{1}}.
\end{equation}
For finite involution points $p$ and $q$ we obtain the following general form
\begin{equation}\label{CUF}
\begin{split}
F_a(u,v)=&G_au^3+H_av^3\\
&+K({u}^{2}va_{{8}}+u{v}^{2}a_{{9}}+{u}^{2}a_{{4}}+uva_{{5}}+{v}^{2}a_{{6}
}+ua_{{2}}+va_{{3}}+a_{{1}})
\end{split}
\end{equation}
We have two Manin involutions,
\begin{equation} \label{2invs}
\iota_p(u,v):=(u,v)+z(c-u,d-v), \qquad \iota_q(u,v):=(u,v)+z(e-u,f-v),
\end{equation}
with $z$ given by (\ref{sec}) and (\ref{expl}), where for the latter involution $(c,d)$ should be replaced by $(e,f)$, and
\begin{align*}
F_a^{(u)}(u,v)&=3G_au^2+K(2uva_{{8}}+{v}^{2}a_{{9}}+2{u}a_{{4}}+va_{{5}}+a_{{2}}),\\
F_a^{(v)}(u,v)&=3H_av^2+K({u}^{2}a_{{8}}+2u{v}a_{{9}}+ua_{{5}}+2{v}a_{{6}}+a_{{3}}), \\
F_a^{(u,u)}(u,v)&=6G_au+2K(va_{{8}}+a_{{4}}),\\
F_a^{(v,v)}(u,v)&=6H_av+2K(ua_{{9}}+a_{{6}}),\\
F_a^{(u,v)}(u,v)&=K(2{u}a_{{8}}+2{v}a_{{9}}+a_{{5}}).
\end{align*}
The involutions (\ref{2invs}) are anti measure-preserving with density $1/F_a(u,v)$.
\begin{proposition}
The composition of the Manin involutions (\ref{2invs}) is an integrable map of the plane. It preserves each curve of the cubic pencil $P_{\alpha,\beta}(u,v)=0$ with (\ref{CUF}) (or (\ref{CUFS})) and it is measure-preserving with density  $1/F_a(u,v)$.
\end{proposition}

Taking $d=e=0$, with $\iota_1=\lim_{c\rightarrow\infty}\iota_{c,0}$ and $\iota_2=\lim_{f\rightarrow\infty}\iota_{0,f}$, the map $\tau=\iota_2\circ\iota_1$ is a special case of the QRT map with
\[
A^0=\begin{pmatrix}
0 & a_8 & a_4 \\
a_9 & a_5 & a_2 \\
a_6 & a_3 & a_1
\end{pmatrix}\quad \text{ and } \quad
A^1=\begin{pmatrix}
0 & b_8 & b_4 \\
b_9 & b_5 & b_2 \\
b_6 & b_3 & b_1
\end{pmatrix}.
\]

\begin{example} \label{cubicexample}
We choose particular values for the constants in $F_a,F_b$ (\ref{CUF}), $a_1=a_9=1$, $a_2=a_3=a_4=-1$, $a_5=a_6=a_8=0$, $b_1=b_9=0$, $b_2=b_3=b_4=-1$, $b_5=b_6=b_8=1$, $c=2$, $d=e=0$, $f=1$. This gives
\begin{equation}\label{Exampl6}
F_a(u,v)=5u^3+8(uv^2-u^2-u-v+1),\ F_b(u,v)=6u^3+8(u^2v-u^2+uv+v^2-u-v).
\end{equation}

\begin{figure}[htb]
\begin{center}
\includegraphics[width=9cm]{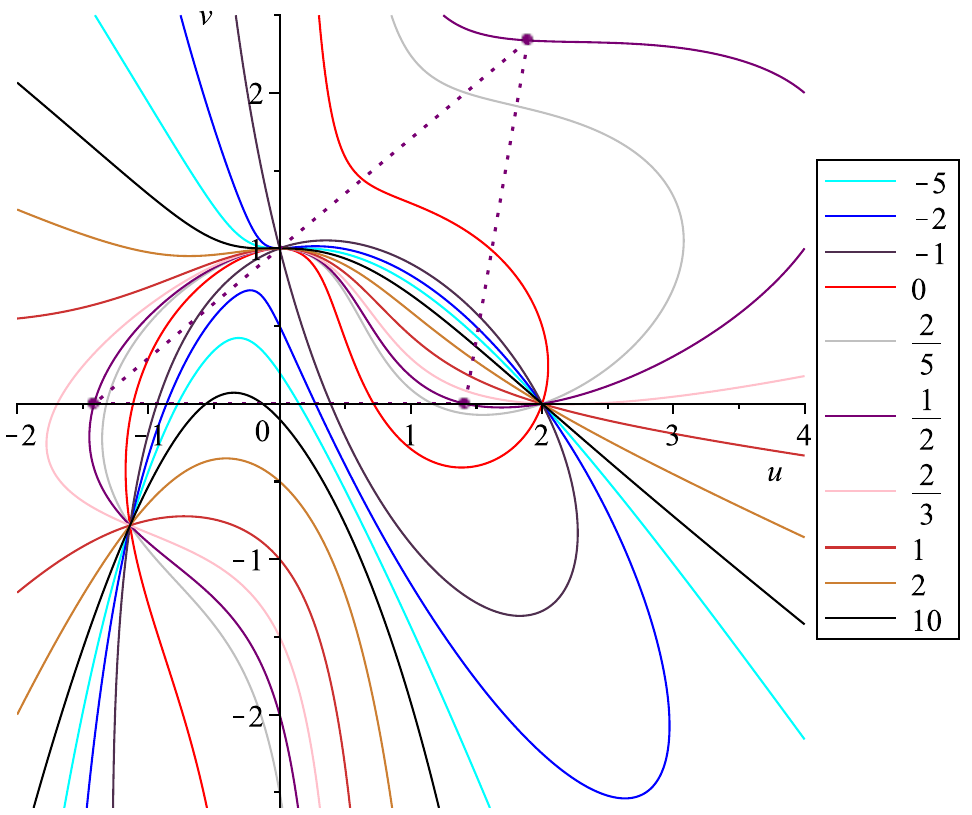}
\caption{\label{CP} Ten curves from the cubic pencil defined by (\ref{pen}) and (\ref{Exampl6}), labeled by the value of $-\beta/\alpha$.}
\end{center}
\end{figure}

In Figure \ref{CP} we have drawn 10 curves of this cubic pencil. In addition to the involution points $(2,0)$ and $(0,1)$ there is one other finite real base point\footnote{Also, there are 4 finite complex base points and all curves on which $(0,1)$ is non-singular are tangent at $(0,1)$.}, near $-(1.140,0.782)$. We have
$
\iota_{2,0}(u,v)=(u,v)-\frac{g}{h}(u-2,v)$, with
\begin{align*}
g&= {u}^{5}+3\,{u}^{4}v+21\,{u}^{3}{v}^{2}+24\,{u}^{2}{v}^{3}+8\,u
{v}^{4}-2\,{u}^{3}v-46\,{u}^{2}{v}^{2}-16\,u{v}^{3}-16\,{v}^{4}\\
&\ \ \ -22\,{u
}^{3}-22\,{u}^{2}v+16\,u{v}^{2}+24\,{v}^{3}+52\,{u}^{2}+16\,uv-16\,{v}
^{2}-24\,u+24\,v-16,\\
h&=
\left( u-2 \right)\left( {u}^{4}+3\,{u}^{3}v+21\,{u}^{2}{v}^{2}+24
\,u{v}^{3}+8\,{v}^{4}-3\,{u}^{3}-7\,{u}^{2}v-44\,u{v}^{2}-8\,{v}^{3}\right.\\
&\ \ \ \left.-6
\,{u}^{2}-8\,uv+4\,{v}^{2}+28\,u+20\,v-24\right),
\end{align*}
and $\iota_{0,1}(u,v)=(u,v)-\frac{k}{l}(u,v-1)$, where
\begin{align*}
k&=
6\,{u}^{5}+{u}^{4}v-11\,{u}^{3}{v}^{2}+8\,{u}^{2}{v}^{3}+8\,u{
v}^{4}-{u}^{4}+44\,{u}^{3}v-8\,{u}^{2}{v}^{2}-16\,u{v}^{3}-33\,{u}^{3}
\\
&\ \ \ -16\,{u}^{2}v+24\,u{v}^{2}+8\,{v}^{3}+16\,{u}^{2}-32\,uv-24\,{v}^{2}+
16\,u+24\,v-8,\\
l&=u \left( 6\,{u}^{4}+{u}^{3}v-11\,{u}^{2}{v}^{2}+8\,u{v}
^{3}+8\,{v}^{4}-{u}^{3}+33\,{u}^{2}v-16\,u{v}^{2}-24\,{v}^{3}-22\,{u}^
{2}\right.\\
&\ \ \ \left.+8\,uv+24\,{v}^{2}-8\,v \right).
\end{align*}
As indicated in the figure, the image of the point $(\sqrt{2},0)$ under the involution $\iota_{2,0}$ is $(-\sqrt{2},0)$, and the image of $(-\sqrt{2},0)$
under $\iota_{0,1}$ is $(\frac{9}{7}+ \frac 37 \sqrt {2} ,\frac {10}{7}+\frac {9}{14}\sqrt {2})$. The image of the curve labeled -1 is the point $(0,1)$ as this is a singular point of that curve.
The Manin line through $(2,0)$ and $(0,1)$ is given by \[
L(u,v)=2-u-2v=0.
\]
In terms of variables
\[
(x,y)=(u,v)/L(u,v)
\]
the involutions $\iota_{2,0}$ and $\iota_{0,1}$ become the horizontal and vertical switches of the QRT map with matrices
\[
A^0=\begin{pmatrix}
0&6&\frac52\\
-2&0&-\frac12\\
-2&-2&-\frac12
\end{pmatrix},\qquad
A^1=\begin{pmatrix}
0&3&4\\
2&4&1\\
2&1&0
\end{pmatrix},
\]
i.e. we have
\begin{align*}
\iota_{2,0}&\mapsto \iota_1: (x,y)\rightarrow
\left(-{\frac { \left( 18\,x{y}^{2}+16\,xy+10\,{y}^{2}+4\,x+5\,y+1 \right)
 \left( 2\,y+1 \right) }{36\,x{y}^{3}+74\,x{y}^{2}+36\,{y}^{3}+35\,xy+
50\,{y}^{2}+9\,x+24\,y+4}}
,y\right),\\
\iota_{0,1}&\mapsto \iota_2: (x,y)\rightarrow
\left(x,{\frac {33\,{x}^{4}-26\,{x}^{3}y-5\,{x}^{3}-28\,{x}^{2}y-14\,{x}^{2}+x
+2\,y+1}{ 2\left( 18\,{x}^{2}y+13\,{x}^{2}+8\,xy+x-2\,y-1 \right)
 \left( x+1 \right) }}\right),
\end{align*}
preserving the ratio of biquadratics
\[
\frac{F_a}{F_b}={\frac {12\,{x}^{2}y-4\,x{y}^{2}+5\,{x}^{2}-4\,{y}^{2}-x-4\,y-1}{2(
3\,{x}^{2}y+2\,x{y}^{2}+4\,{x}^{2}+4\,xy+2\,{y}^{2}+x+y)}}.
\]
\end{example}

\section{Quartic pencils} \label{squar}
In this section, pencils of degree $N=4$ are considered. With $p$ and $q$ double base points, the 22-parameter map $\tau=\iota_{q}\circ\iota_{p}$ is measure-preserving with density $L(u,v)/F_a(u,v)$.

Let the quartic curve $F_a(u,v)=0$, with
\begin{align*}
F_a(u,v)&:=a_1 + a_2 u + a_3 v + a_4 u^2 + a_5 u v + a_6 v^2 + a_7 u^3 + a_8 u^2 v + a_9 u v^2 +a_{10} v^3\\
&\ \ \ \ +a_{11} u^4 + a_{12} u^3 v + a_{13} u^2 v^2 +a_{14} uv^3 +a_{15} v^4
\end{align*}
have double points at $p=(c,d)$ and $q=(e,f)$, i.e. at these points we require the function $F_a$ as well as its first partial derivatives $F_a^{(u)}$, $F_a^{(v)}$ to vanish. Generically the genus of such a curve is 1, the same as in the cubic case. Assuming that
\[
V:=c^3f^3-d^3e^3\neq0,\quad
W:=(cf-de)^2((cf+de)^2+2cdef)\neq0,
\]
we can solve for
\[
a_{{7}}=\frac PV,\
a_{{10}}=\frac QV,\
a_{{11}}=\frac R{VW},\
a_{{12}}=\frac S{VW},\
a_{{14}}=\frac T{VW},\
a_{{15}}=\frac U{VW},
\]
where the functions $P,Q,R,S,T,U$ can be found in Appendix D. If $V$ or $W$ vanishes one has to solve for other parameters. If $c\neq e$ one can solve for $a_1,a_2,a_3,a_4,a_5,a_7$ and if $d\neq f$ one can solve for $a_1,a_2,a_3,a_5,a_6,a_{10}$. The parameters $a_7,a_{10},a_{11},a_{12}$, $a_{14},a_{15}$ vanish when $d=e=0$ in the limit where both $c$ and $f$ go to infinity, leaving us with the most general biquadratic. For finite $p$ and $q$, we obtain
\begin{equation}\label{quart}
\begin{split}
F_a(u,v) &= \left( {u}^{2}{v}^{2}a_{{13}}+{u}^{2}va_{{8}}+u{v}^{2}a_{{9}}+{u}^{2}
a_{{4}}+uva_{{5}}+{v}^{2}a_{{6}}+ua_{{2}}+va_{{3}}+a_{{1}} \right) WV \\
&\ \ \ +
 \left( P{u}^{3}+Q{v}^{3} \right) W+{u}^{4}R+{u}^{3}vS+u{v}^{3}T+{v}^{4}U.
\end{split}
\end{equation}
As in the previous section, we have two involutions,
\begin{equation} \label{2invs2}
\iota_{p}(u,v):=(u,v)+z(c-u,d-v), \qquad \iota_{q}(u,v):=(u,v)+z(e-u,f-v).
\end{equation}
Here $z$ is again given by (\ref{sec}) and (\ref{expl}), where for the second involution $(c,d)$ should be replaced by $(e,f)$, but now
\begin{align*}
F_a^{(u)}(u,v)&=\left( 2{u}{v}^{2}a_{{13}}+2{u}va_{{8}}+{v}^{2}a_{{9}}+2{u}
a_{{4}}+va_{{5}}+a_{{2}}\right) WV \\
&\ \ \ + 3{u}^{2} PW+ 4{u}^{3}R+3{u}^{2}vS+{v}^{3}T,\\
F_a^{(v)}(u,v) &= \left( 2{u}^{2}{v}a_{{13}}+{u}^{2}a_{{8}}+2u{v}a_{{9}}+ua_{{5}}+2{v}a_{{6}}+a_{{3}} \right) WV \\
&\ \ \ +{v}^{3} QW+{u}^{3}S+3u{v}^{2}T+4{v}^{3}U,\\
F_a^{(u,u)}(u,v)&=\left( 2{v}^{2}a_{{13}}+2va_{{8}}+2a_{{4}}\right) WV + 6{u}PW+ 12{u}^{2}R+6{u}vS,\\
F_a^{(v,v)}(u,v) &= \left( 2{u}^{2}a_{{13}}+2ua_{{9}}+2a_{{6}} \right) WV +3{v}^{2}QW+6uvT+12{v}^{2}U,\\
F_a^{(u,v)}(u,v)&=\left( 4{u}{v}a_{{13}}+2{u}a_{{8}}+2{v}a_{{9}}+a_{{5}}\right) WV +3{u}^{2}S+3{v}^{2}T.\\
\end{align*}

Both involutions are anti measure-preserving, $\left(s_1(u-c)+s_2(v-d)\right)/F_a(u,v)$ is the density for $\iota_{p}$ and $\left(t_1(u-e)+t_2(v-f)\right)/F_a(u,v)$ is the density for $\iota_{q}$, where $s_i,t_i$ are arbitrary.
\begin{proposition}
The composition of the generalised Manin involutions (\ref{2invs2}) is an integrable map of the plane. It preserves the quartic pencil $P_{\alpha,\beta}(u,v)=0$ with (\ref{quart}), and it is measure-preserving with density  $L(u,v)/F_a(u,v)$,
where $L=0$ is the Manin line (\ref{L}).
\end{proposition}

\begin{example} \label{quarticexample}
Consider the quartic pencil where
\begin{equation}\label{Exampl8a}
F_a(u,v)={u}^{2} \left( 28\,u^2-24\,uv+12\,v^2+16\,u-8\,v-7 \right)
\end{equation}
is a product of a double line and an ellipse, and
\begin{equation}\label{Exampl8b}
F_b(u,v)=  \left(u -3\,v \right)  \left(2\,u + v-1\right)  \left(3\,u + v \right)  \left(u+ 5\,v-5 \right).
\end{equation}
is a product of four lines. All 10 base points are finite, the involution points are the singular base points $(0,0)$ and $(0,1)$. Some curves of the pencil are plotted in Figure \ref{F1}.

\newpage

\begin{figure}[htb]
\begin{center}
\includegraphics[width=9.5cm]{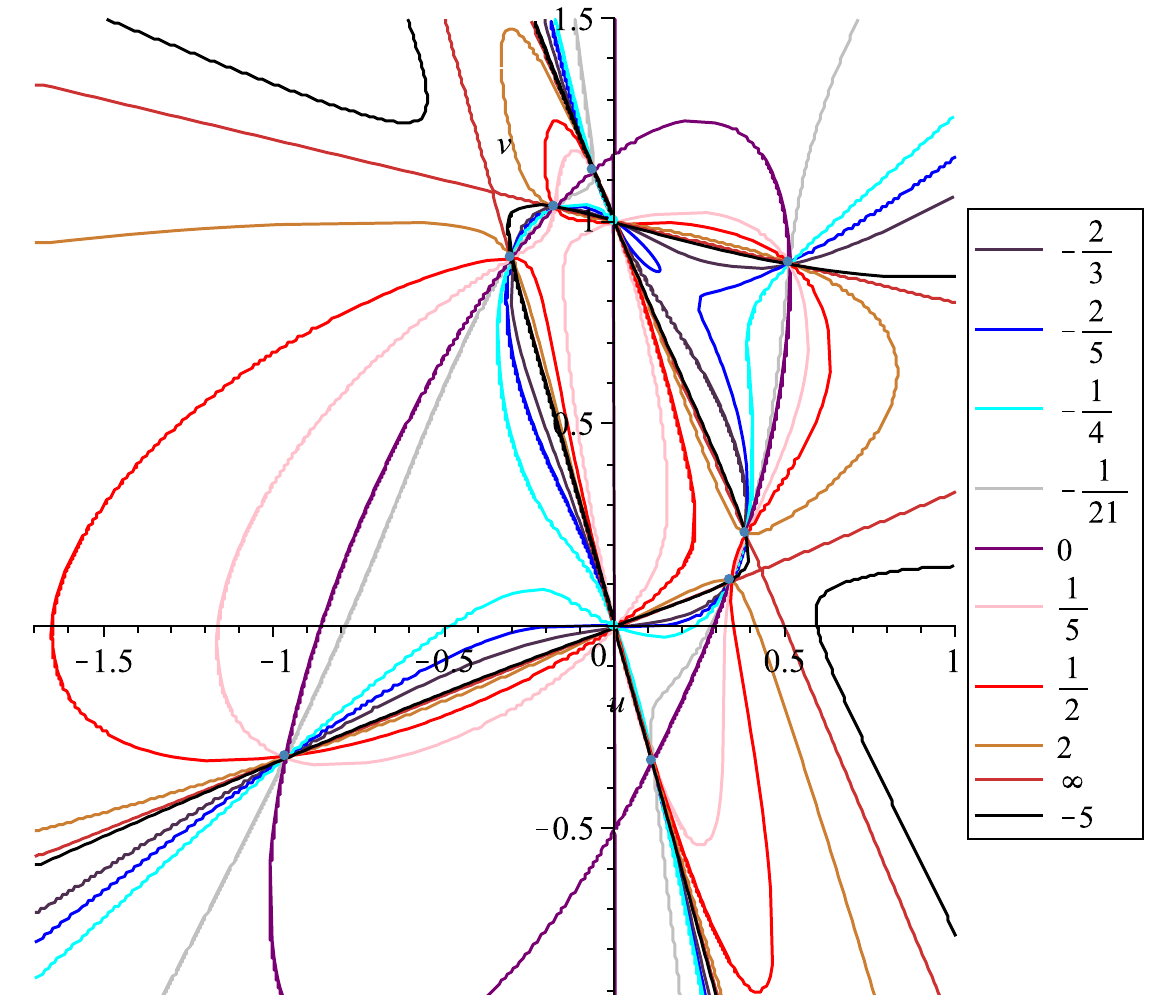}
\end{center}
\caption{\label{F1} Ten curves from the quartic pencil defined by (\ref{pen}), (\ref{Exampl8a}) and (\ref{Exampl8b}), labeled by  $-\beta/\alpha$.}
\end{figure}

The curve which contains the point $(-\frac32,\frac 3{10})$ and some of its iterates are plotted in Figure \ref{F2}.

\begin{figure}[htb]
\begin{center}
\includegraphics[width=7cm]{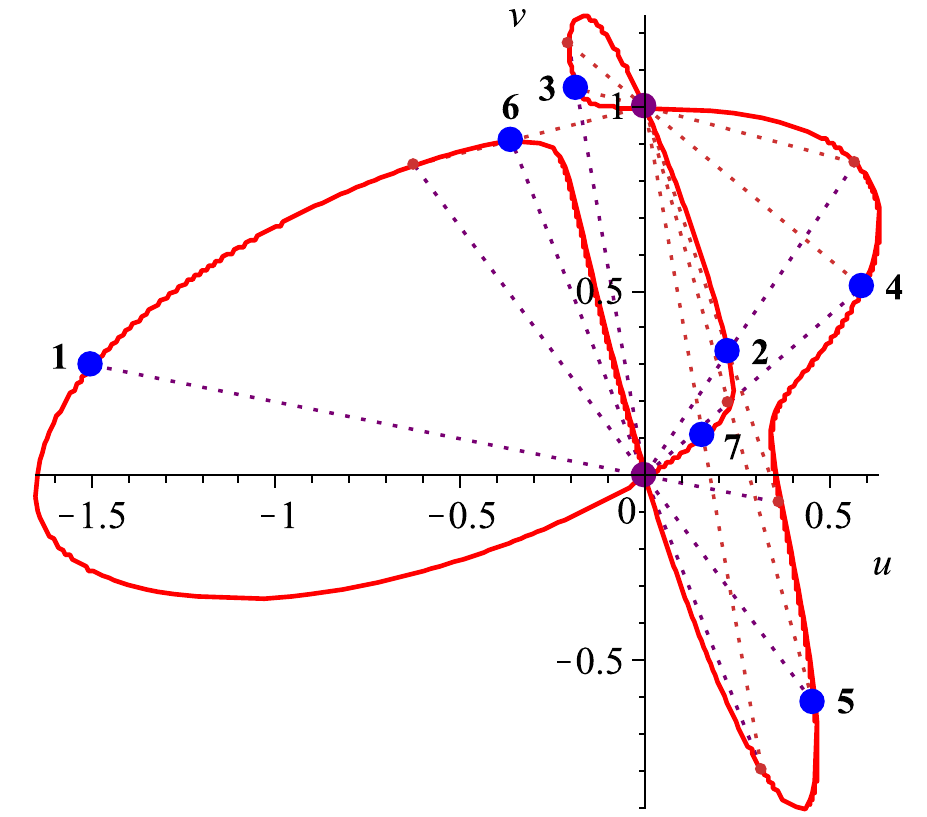}
\end{center}
\caption{\label{F2} Six iterations of the point $(-\frac32,\frac 3{10})$ under the Manin transformation (\ref{MI}), $\iota_{0,1}\circ\iota_{0,0}$.}
\end{figure}

The involutions are explicitly given by:
\begin{equation} \label{MI}
\iota_{0,0}(u,v)=(u,v)A,\qquad \iota_{0,1}(u,v)=(0,1)-3(u,v-1)B
\end{equation}
with
\[
A={\frac {154\,{u}^{2}-43\,uv+95\,{v}^{2}+3\,u-110\,v}{340\,{u}^{3}+176
\,{u}^{2}v-116\,u{v}^{2}+80\,{v}^{3}-154\,{u}^{2}+43\,uv-95\,{v}^{2}}},
\]
and
\[
B={\frac {25\,{u}^{2}-16\,uv+15\,{v}^{2}+16\,u-8\,v-7}{200\,{u}^{3}+
88\,{u}^{2}v-152\,u{v}^{2}+24\,{v}^{3}-13\,{u}^{2}+256\,uv-27\,{v}^{2}
-104\,u-18\,v+21}}.
\]

The set of base points of the pencil is the union of the disjoint sets of points where $A$ resp. $B$ are undefined. This is made clear in Figure \ref{F9}.

\begin{figure}[htb]
\begin{center}
\includegraphics[width=6.5cm]{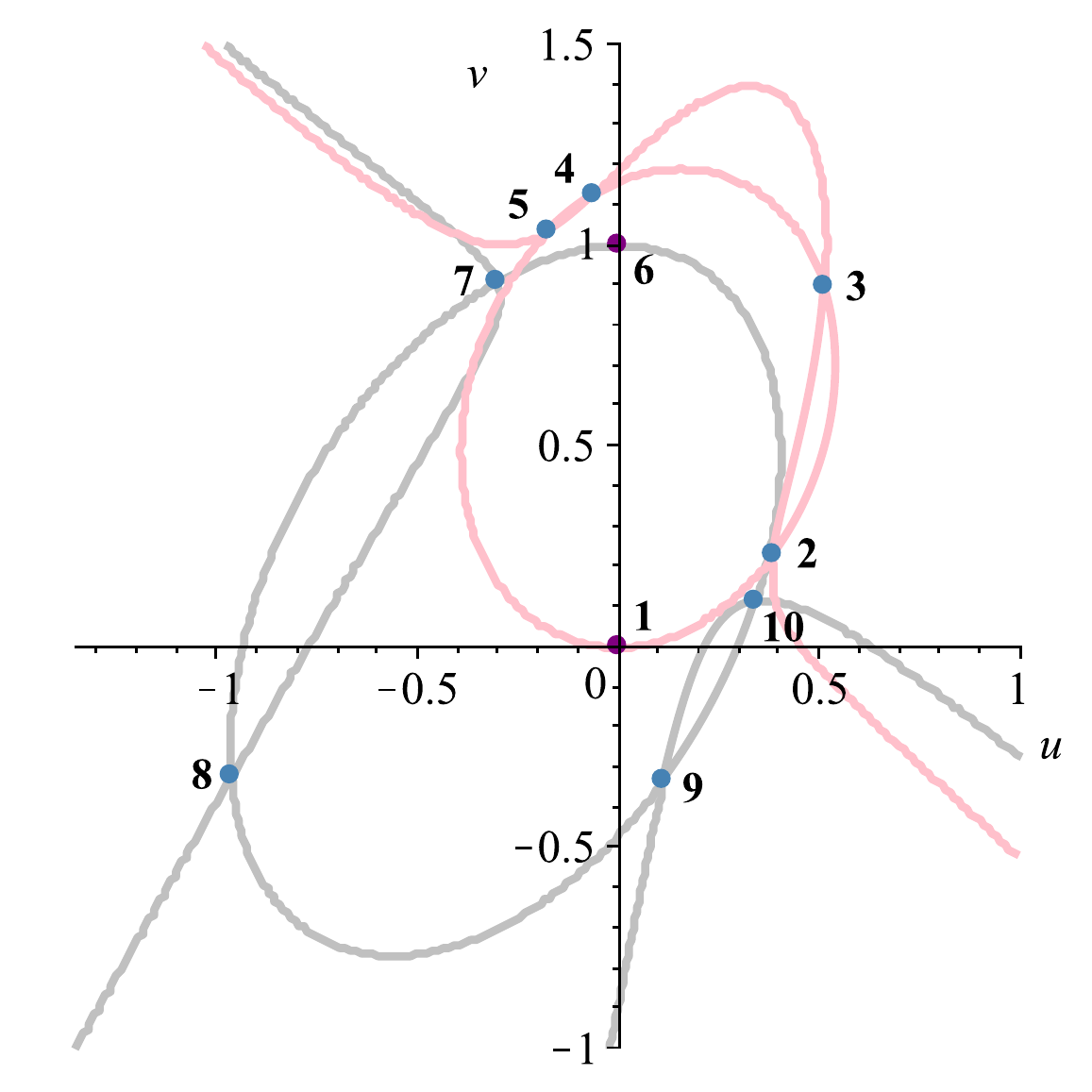}
\end{center}
\caption{\label{F9} The base points lie on curves defined by the numerators and denominators of $A$ (pink) and $B$ (grey).}
\end{figure}

We have $\iota_{0,1}(b_2)=b_4$, $\iota_{0,0}(b_7)=b_9$,
$\iota_{0,1}(b_3)=b_5$, and $\iota_{0,0}(b_8)=b_{10}$. To define the action of  $\iota_{0,0}$ at $b_2$, $b_3$, $b_4$ and $b_5$, one needs to blow up at these points. Similarly, for $\iota_{0,1}$ blow-ups are required at $b_7$, $b_8$, $b_9$ and $b_{10}$.

Performing a change of variables, $(x,y)=(1-u-v,-v)/u$, the involutions become
\begin{align*}
(x,y)&\rightarrow \left({\frac {15\,{y}^{2}x-150\,yx-3\,{y}^{2}-157\,x-58\,y+29}{110\,yx-15\,{
y}^{2}+3\,x+150\,y+157}},y\right)\\
(x,y)&\rightarrow \left(x,-{\frac {21\,{x}^{2}y-56\,{x}^{2}-6\,yx-50\,x-102\,y+206}{21\,{x}^{2}-
66\,yx-6\,x-66\,y-102}}
\right).
\end{align*}
preserving the ratio of biquadratics
\[
\frac{F_a}{F_b}=\frac {7\,{x}^{2}-22\,yx+3\,{y}^{2}-2\,x-30\,y-37}{(x-1)(5x+4)(y-3)(3y+1)},
\]
i.e. we obtain the QRT map with matrices
\[
A^0=\begin{pmatrix}
0&0&7\\
0&-22&-2\\
3&-30&-37
\end{pmatrix},\qquad
A^1=\begin{pmatrix}
15&-40&-15\\
-3&8&3\\
-12&32&12
\end{pmatrix}.
\]
\end{example}

The special involutions with base points at infinity, with $d=e=0$,
\[
\iota_1=\lim_{c\rightarrow\infty}\iota_{c,0}, \qquad \iota_2=\lim_{f\rightarrow\infty}\iota_{0,f}
\]
are anti measure-preserving. The horizontal switch $\iota_1$ has $\left(s_1v+s_2\right)/F_a(u,v)$ as density, and the vertical switch $\iota_2$ has density
$\left(t_1u+t_2\right)/F_a(u,v)$, for arbitrary $s_i,t_i$. This implies that $\iota_{c,d}\circ \iota_1$ is measure-preserving with density
$
\left(v-d\right)/F_a(u,v)$
and, that $\tau=\iota_2\circ \iota_1$ is measure-preserving with density $1/F_a(u,v)$.
This map $\tau$ is the QRT map.

\section{Roots of generalised Manin transformations} \label{roots}
In this section, we specify subfamilies of generalised Manin transformations which admit a root, i.e. maps that can be written as $\tau=\rho^2$, such as the 12-parameter symmetric QRT map.

Recall that the QRT map is obtained by considering a $N=4$ pencil with double base points at $(0,z)$ and
$(z,0)$ as involution points, and taking the limit where $z\rightarrow\infty$. In that limit the quartic polynomials
$F_a(u,v)$ and $F_b(u,v)$ become biquadratic polynomials.
A special case of the QRT map, the so called symmetric QRT map, arises when the biquadratic polynomials are symmetric in $u,v$, i.e. they are invariant under what Duistermaat calls the {\em symmetry switch} \cite[Section 10.1]{Dui}
\begin{equation} \label{dsw}
\sigma(u,v)=(v,u).
\end{equation}

The symmetric QRT map $\tau=\iota_2\circ\iota_1$ equals $\tau=\rho^2$, where $\rho=\sigma \circ \iota_1 = \iota_2 \circ \sigma$ is called the QRT-root.

We note that $\sigma$ may arise as a Manin involution corresponding to the base point $(z,-z)$ in the limit where $z\rightarrow\infty$, and we provide an example of a map which can be written as a Manin transformation in various different ways.
\begin{example}
The Lyness map
\[
\lambda:(u,v)\rightarrow \left(v,\frac{v+a}{u}\right)
\]
leaves invariant the pencil of cubic curves
\[
\alpha(u+1)(v+1)(u+v+a)+\beta uv=0.
\]
The pencil has finite base points $p_1=(-1,0)$, $p_2=(0,-1)$, $p_3=(-a,0)$, $p_4=(0,-a)$, which gives rise to involutions
\begin{align*}
\iota_{p_1}(u,v)&=\left(\frac{a(u+1)+v}{uv},\frac{a+v}{u}\right),\\
\iota_{p_2}(u,v)&=\left(\frac{a+u}{v},\frac{u+a(v+1)}{uv}\right),\\
\iota_{p_3}(u,v)&=\left(\frac{u+a(v+1)}{uv},\frac{a(uv+v+1)+u}{u(u+a)}\right),\\
\iota_{p_4}(u,v)&=\left(\frac{a(uv+u+1)+v}{v(v+a)},\frac{v+a(u+1)}{uv}\right),
\end{align*}
as well as base points at infinity $p_5=\lim_{x\rightarrow \infty}(0,x)$,
$p_6=\lim_{x\rightarrow \infty}(x,0)$ (these have multiplicity two), and $p_7=\lim_{x\rightarrow \infty}(x,-x)$, which yield the involutions
\[
\iota_{p_5}(u,v)=\left(u,\frac{a+u}{v}\right),\quad
\iota_{p_6}(u,v)=\left(\frac{a+v}{u},v\right),\quad
\iota_{p_7}(u,v)=(v,u).
\]

The latter Manin involution, $\iota_{p_7}=\sigma$, is the symmetry switch of the pencil of curves, it is a reversing symmetry for the Lyness map, and it corresponds to negation in the group law of the cubic \cite{BR}. The other involutions are also reversing symmetries, generated by $\lambda$ and $\sigma$:
\[
\iota_{p_1}=\sigma\circ \lambda^2,\
\iota_{p_2}=\lambda^2 \circ \sigma,\
\iota_{p_3}=\lambda^3 \circ \sigma,\
\iota_{p_4}=\sigma \circ \lambda^3,\
\iota_{p_5}=\lambda \circ \sigma,\
\iota_{p_6}=\sigma \circ \lambda.
\]
Thus the Lyness map is a QRT root: we have $\iota_{p_5}=\iota_2$ and $\iota_{p_6}=\iota_1$, see (\ref{sevol}), and hence
\[
\lambda = \sigma \circ \iota_1=\iota_2 \circ \sigma.
\]
On the other hand, it can also be written as the composition of two Manin involutions which correspond to finite involution points
\[
\lambda=\iota_{p_1}\circ\iota_{p_4}
=\iota_{p_3}\circ\iota_{p_2},
\]
or as the composition of a Manin involution which corresponds to a finite involution point and a horizontal or vertical switch
\[
\lambda=\iota_{p_2} \circ \iota_2=\iota_1 \circ \iota_{p_1}.
\]
\end{example}

In the sequel we call a transformation $\sigma$ a symmetry switch of the pencil $P=0$ if $\sigma$ is a symmetry of
$P$ and it is an involution.

\begin{theorem} \label{theroot}
Let $\sigma$ be a symmetry switch of the pencil $P_{\alpha,\beta}(u,v)=0$ which maps lines to lines. Then
\[
\tau_p = \iota_{\sigma(p)}\circ \iota_p = \rho_p^2,\quad \text{ with } \rho_p=\sigma \circ \iota_p=\iota_{\sigma(p)} \circ \sigma.
\]
We call $\rho_p$ the root of $\tau_p$.
\end{theorem}

\begin{proof}
Let $q$ be a point on a curve $C$ in a pencil of degree $N$, and let the involution point $p$ be a singular point of multiplicity $N-2$. Note that $\sigma(p)$ has the same multiplicity as $p$. Defining $r=\iota_p(q)\in C$, the points $p,q,r$ are collinear. Because $\sigma$ maps lines to lines the points $\sigma(p),\sigma(q),\sigma(r)$ are also collinear. Because $\sigma$ is a symmetry, both $\sigma(q),\sigma(r)$ are on the curve $C$. Therefore we must have
$\sigma(r)=\iota_{\sigma(p)}(\sigma(q))$, cf. Figure \ref{FS}. And hence $\tau_p = \iota_{\sigma(p)}\circ \iota_p = \iota_{\sigma(p)}\circ \sigma^2 \circ \iota_p = \sigma \circ \iota_p \circ \sigma \circ \iota_p = \rho_p^2$.
\end{proof}

As Theorem \ref{theroot} concerns symmetry switches which map lines to lines, it would be worthwhile to determine which projective collineations are symmetry switches and to study the corresponding pencils. In the next subsections we consider the symmetric case, and we introduce a more general linear symmetry switch. In Appendix E we show that the highest dimensional solution yields pencils comprising singular curves only.

\noindent

\subsection{Symmetric generalised Manin transformations}
We require that the symmetric quartic polynomials $F_a$ and $F_b$, where
\begin{align*}
F_a&=a_{{1}}+a_{{2}} \left( u+v \right) +uva_{{3}}+ \left( {u}^{2}+{v}^{2}
 \right) a_{{4}}+ \left( {u}^{2}v+u{v}^{2} \right) a_{{5}}+{u}^{2}{v}^
{2}a_{{6}}\\
&\ \ \ + \left( {u}^{3}+{v}^{3} \right) a_{{7}}+ \left( {u}^{3}v+u{
v}^{3} \right) a_{{8}}+ \left( {u}^{4}+{v}^{4} \right) a_{{9}},
\end{align*}
have a singular point at $p=(c,d)$. Solving the constraints for $F_a$ for $a_7,a_8,a_9$ gives
\begin{align*}
a_7&=-\frac{4\,a_{{1}}+ \left( 3\,c+3\,d \right) a_{{2}}+2\,cda_{{3}}+ \left( 2\,{
c}^{2}+2\,{d}^{2} \right) a_{{4}}+ \left( {c}^{2}d+c{d}^{2} \right) a_
{{5}}
}{
\left( c+d \right)  \left( {c}^{2}-cd+{d}^{2} \right)}\\
a_8&=-\frac{1}{
\left( {c}^{2}-cd+{d}^{2} \right)  \left( {c}^{4}+4\,{c}^{2}{d}^{2}+{
d}^{4} \right)  \left( c+d \right) ^{2}
}\Big(
-12\,{c}^{2}{d}^{2}a_{{1}}+ \left( {c}^{5}+d{c}^{4}-8\,{d}^{2}{c}^{3}\right. \\
&\ \ \ \left.-
8\,{d}^{3}{c}^{2}+{d}^{4}c +{d}^{5} \right) a_{{2}}+ \left( {c}^{6}+d{c
}^{5}+{d}^{2}{c}^{4}-4\,{d}^{3}{c}^{3}+{d}^{4}{c}^{2}+{d}^{5}c+{d}^{6}
 \right) a_{{3}} \\
 &\ \ \ + \left( 2\,d{c}^{5}-4\,{d}^{2}{c}^{4} -4\,{d}^{4}{c}^{
2}+2\,{d}^{5}c \right) a_{{4}}+ \left( {c}^{7}+3\,d{c}^{6}+3\,{d}^{2}{
c}^{5}+{d}^{3}{c}^{4}+{d}^{4}{c}^{3} \right.\\
&\ \ \ \left. +3\,{d}^{5}{c}^{2}+3\,{d}^{6}c+{d}
^{7} \right) a_{{5}}+ \left( 2\,{c}^{7}d+2\,{c}^{6}{d}^{2} +2\,{c}^{5}{
d}^{3}+4\,{c}^{4}{d}^{4}+2\,{c}^{3}{d}^{5}+2\,{c}^{2}{d}^{6}\right.\\
&\ \ \ \left. +2\,c{d}^{
7} \right) a_{{6}}\Big)
\\
a_9&=\frac{1}{\left( {c}^{2}-cd+{d}^{2} \right)  \left( {c}^{4}+4\,{c}^{2}{d}^{2}+{
d}^{4} \right)  \left( c+d \right) ^{2}
}\Big(\left( 3\,{c}^{4}+3\,{c}^{3}d+12\,{c}^{2}{d}^{2}+3\,c{d}^{3}\right.\\
&\ \ \ \left. +3\,{d}^{
4} \right) a_{{1}} + \left( 2\,{c}^{5}+5\,d{c}^{4}+11\,{d}^{2}{c}^{3}+
11\,{d}^{3}{c}^{2}+5\,{d}^{4}c+2\,{d}^{5} \right) a_{{2}}+ \left( 2\,d
{c}^{5}\right.\\
&\ \ \ \left.+2\,{d}^{2}{c}^{4}+6\,{d}^{3}{c}^{3}+2\,{d}^{4}{c}^{2}+2\,{d}^{
5}c \right) a_{{3}}+ \left( {c}^{6}+d{c}^{5}+7\,{d}^{2}{c}^{4}+2\,{d}^
{3}{c}^{3}+7\,{d}^{4}{c}^{2}\right.\\
&\ \ \ \left. +{d}^{5}c+{d}^{6} \right) a_{{4}}+ \left(
d{c}^{6}+3\,{d}^{2}{c}^{5}+4\,{d}^{3}{c}^{4}+4\,{d}^{4}{c}^{3}+3\,{d}^
{5}{c}^{2}+{d}^{6}c \right) a_{{5}}+ \left( {c}^{6}{d}^{2}\right.\\
&\ \ \ \left.+{c}^{5}{d}^
{3}+{c}^{3}{d}^{5}+{c}^{2}{d}^{6} \right) a_{{6}}\Big)
\end{align*}
and similar expressions are obtained for $b_7,b_8,b_9$. Taking $\sigma(u,v)=(v,u)$, one defines $\rho_p = \sigma \circ \iota_p$ and verifies that $\rho_p = \iota_{\sigma(p)} \circ \sigma$. The symmetric QRT-root is obtained by considering the limit $d\rightarrow\infty$ (in which $a_7,a_8,a_9,b_7,b_8,b_9\rightarrow 0$), or by performing a fractional affine transformation explained in section \ref{QRTform}.

One can also solve the constraints for other variables, depending on what variables one chooses to be non-zero
\begin{example} \label{symex}
Setting $a_4=1,a_3=a_5=a_6=a_7=a_8=0$ and $b_3=1,b_4=b_5=b_6=b_7=b_8=0$, both polynomials $F_a$ and $F_b$ have singular points at both $(0,1)$
and $(1,0)$ if
\[
a_1=a_9=-\frac{1}{2},a_2=0,b_1=\frac34,b_2=-1,b_9=\frac14.
\]
Thus we obtain the map
\[
(u,v)\rightarrow (v,u)-2\,{\frac { {u}^{4}+{v}^{4}-2\,{u}^{3}+2\,u-1 }{{u}^{4}+{v}^{4}-4\,{u}^{3}+6\,{u}^{2}-4\,u+1}}(v,
u-1),
\]
which preserves the pencil
\[
 \alpha\left( {u}^{4}+{v}^{4}-2({u}^{2}+{v}^{2})+1 \right)+ \beta
 \left( {u}^{4}+{v}^{4}+4(uv-u-v)+3 \right)
=0.
\]
After the transformation, $(x,y)=(u,v)/(1-u-v)$, the map becomes the composition of $(x,y)\rightarrow(y,x)$ and the horizontal switch which preserves the ratio of biquadratics
\[
\frac{F_a}{F_b}={\frac {2\,{x}^{2}{y}^{2}+8\,{x}^{2}y+8\,x{y}^{2}+4\,{x}^{2}+12\,xy+4
\,{y}^{2}+4\,x+4\,y+1}{2\,{x}^{2}{y}^{2}+8\,{x}^{2}y+8\,x{y}^{2}+6\,{x
}^{2}+16\,xy+6\,{y}^{2}+8\,x+8\,y+3}}.
\]
\end{example}

\subsection{Linear symmetry switches}
We introduce a symmetry switch that is more general than (\ref{dsw}), but which is still linear.
In terms of
\[
U=(u,v),\ V=(b,-a),\ W=(ad-bc,ae-bd),\ E = V\cdot W,\ G = G(U) = U\cdot W
\]
we define
\begin{equation} \label{lsw}
\sigma_{a,b,c,d,e}:U\rightarrow U-\frac{2G}{E}V.
\end{equation}
The `symmetric switch' given by (\ref{dsw}) is a special case of (\ref{lsw}), we have $\sigma=\sigma_{a,a,c,d,c}$ and the matrices of $\sigma$ and $\sigma_{a,b,c,d,e}$ are conjugate. In the sequel we will omit the index $_{a,b,c,d,e}$. The linear transformation $\sigma$ given by (\ref{lsw}) is a reflection in the line through $(0,0)$ perpendicular to $W$
along a line with direction $V$, i.e. we have
\[
\sigma(V)=-V,\qquad \sigma(JW)=JW, \qquad J=\begin{pmatrix} 0 & 1 \\ -1 & 0 \end{pmatrix}.
\]

Importantly, $\sigma$ (\ref{lsw}) leaves the linear respectively quadratic forms
\[
L=L(U)=au+bv,\qquad Q=Q(U)=cu^2+2duv+ev^2
\]
invariant (and it also negates the linear form $G$), that is
\[
L(\sigma(U))=L(U),\qquad Q(\sigma(U))=Q(U), \qquad G(\sigma(U))=-G(U).
\]

For $N=2$ the most general pencil which admits $\sigma$ (\ref{lsw}) as a symmetry is given by
\begin{equation} \label{fafb}
F_a=a_1+a_2L+a_3L^2+a_4Q,\ F_b=b_1+L+L^2+Q.
\end{equation}
Note that the constants $b_2,b_3,b_4$ can be absorbed by the other constants,
\begin{align*}
(a,b)&\rightarrow \frac{1}{b_2} (a,b) \\
(c,d,e)&\rightarrow \frac{1}{b_4} (c,d,e) + \left(1-\frac{b_3}{b_2^2}\right)\frac{1}{b_4} (a^2,ab,b^2).
\end{align*}
We are still free to choose the coordinates of $p$, so in total the degree $N=2$ family of maps which admit a root
has 12 parameters.

\begin{proposition} \label{quadroot}
The root $\rho_p=\sigma \circ \iota_p$, where $\sigma$ is given by (\ref{lsw}) and $\iota_p$ by (\ref{evol}), is an integrable map of the plane. It preserves each curve of the quadratic pencil $P_{\alpha,\beta}(u,v)=0$ with (\ref{pen}) and (\ref{fafb}), and it is measure-preserving with density  $(F_a(U)(L(U)-L(p))^{-1}$.
\end{proposition}

\begin{example} \label{liswex}
Let $(a,b,c,d,e)=(1,2,-3,4,5)$, $(a_1,a_2,a_3,a_4)=(1,-2,-3,4)$, and $b_1=1$. Then
\begin{equation} \label{suv}
\sigma(u,v)=\frac{1}{23}\begin{pmatrix}
-17 & 12\\
20 & 17
\end{pmatrix}
\begin{pmatrix} u \\ v \end{pmatrix}
\end{equation}
and
\[
F_a=-15\,{u}^{2}+20\,uv+8\,{v}^{2}-2\,u-4\,v+1,\
F_b=-2\,{u}^{2}+12\,uv+9\,{v}^{2}+u+2\,v+1.
\]
The point $s=(1/2,-1)$ is on the curve
\begin{equation}\label{Exampl13}
0=P_{8,7}(u,v)=-134\,{u}^{2}+244\,uv+127\,{v}^{2}-9\,u-18\,v+15.
\end{equation}
Choosing $p=(2,-1)$ we find $r=\iota_p(s)=(-160/67,-1)$. The points
\[
\sigma(p)=(-2,1),\quad
\sigma(s)=\left(-\frac{41}{46},-\frac{7}{23}\right),\quad
\sigma(r)=\left(\frac{1916}{1541},-\frac{4339}{1541}\right)
\]
are collinear, and
\[
\iota_p(\sigma(r))=\left(\frac{259627}{86963}, \frac{118690}{86963}\right),\quad
\iota_{\sigma(p)}(r)=\left(-\frac{5651}{3781}, \frac{13630}{3781}\right).
\]
It can be seen, see Figure \ref{FS}, that $\sigma(\iota_p(\sigma(r)))=\iota_{\sigma(p)}(r)$.

\begin{figure}[htb]
\begin{center}
\includegraphics[width=9cm]{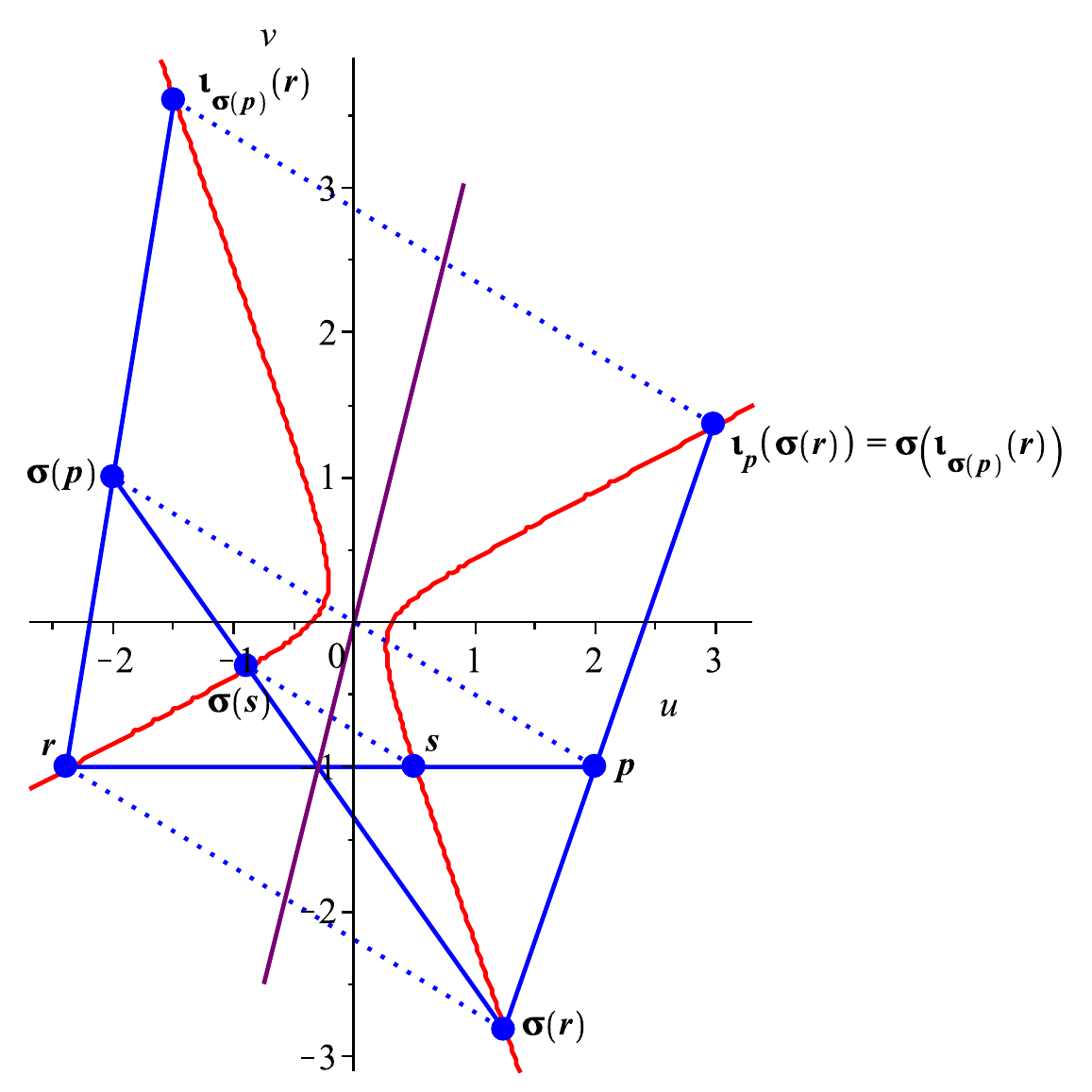}
\end{center}
\caption{\label{FS} A degree 2 curve, given by (\ref{Exampl13}), which admits the symmetry switch (\ref{suv}).
The symmetry switch is a reflection in the line through $(0,0)$ perpendicular to $W=(10,-3)$ (purple), in the direction $(2,-1)$ (dotted).}
\end{figure}

After a transformation, with new coordinates
\[
(x,y)=\left(-{\frac {3\,u-23+29\,v}{2\,u+4\,v}},-23\,{\frac {u-1+v}{2\,u+4\,v}}\right),
\]
we have that $\sigma$ switches $x$ and $y$, we have $\iota_{p}\rightarrow\iota_1: (x,y)\rightarrow \left(f(x,y),y\right)$, $ \iota_{\sigma(p)}\rightarrow\iota_2: (x,y)\rightarrow \left(x,f(y,x)\right)$, where
\[
f(x,y)={\frac {12\,{y}^{3}-213\,xy+651\,{y}^{2}-5966\,x+12084\,y-3268}{12\,xy
+213\,x+213\,y+5966}},
\]
and the preserved ratio is
\[
\frac{F_a}{F_b}={\frac {91\,{x}^{2}-186\,xy+91\,{y}^{2}+20\,x+20\,y-836}{22\,{x}^{2}-
48\,xy+22\,{y}^{2}-49\,x-49\,y-1710}}
.
\]
\end{example}

For $N=3$ the most general pencil left invariant by $\sigma$ (\ref{lsw}) is
\[
F_a=a_1+a_2L+a_3L^2+a_4Q+a_5L^3+a_6LQ,\ F_b=b_1+L+L^2+Q+b_5L^3+b_6LQ.
\]
We require that the involution point $p$ is a point on both $F_a=0$ and $F_b=0$ and
thus we have a 14 parameter family of maps which admit a root. In the cubic case the root
is measure-preserving with density $1/F_a(u,v)$.

For $N=4$ the most general pencil invariant under $\sigma$ (\ref{lsw}) is defined by
\begin{align*}
F_a&=a_1+a_2L+a_3L^2+a_4Q+a_5L^3+a_6LQ+a_7L^4+a_8L^2Q+a_9Q^2,\\
F_b&=b_1+L+L^2+Q+b_5L^3+b_6LQ+b_7L^4+b_8L^2Q+b_9Q^2.
\end{align*}
Here we require that the involution point $p$ is a double point of $F_a=0$ and $F_b=0$,
which gives 6 constraints. Thus we are left with a 16-parameter family whose square root can be taken.
In the quartic case the root is measure-preserving with density $(L(U)-L(p))/F_a(U)$.

In \cite{CMMOQ2,KCMMOQ} it was shown that the Kahan discretisation for several classes of ODE systems of the form
\def\X{{\bf x}}
\[
 \frac{d}{dt}\begin{pmatrix}x\\y\end{pmatrix} = \varphi(x,y) \begin{pmatrix}
0 & 1  \\
-1 & 0 \end{pmatrix}  \nabla H(x,y)
\] and $\varphi(x,y)$ and $H(x,y)$ are scalar functions, can be geometrically understood as the root of a generalised Manin transformation. These classes of ODE systems include physical applications such as: a two-dimensional sub-system of the three-dimensional non-holonomic Suslov problem which describes the motion of a rigid body under the constraint that a certain component of the angular velocity vector vanishes, the reduced Nahm equations \cite{Nahm} corresponding to tetrahedrally symmetric monopoles of charge 3, and the reduced Nahm equations for octahedrally symmetric monopoles of charge 4.

For canonical Hamiltonian system with cubic $H$, it was shown in \cite{PSS} that the Kahan  map can be represented
in six different ways as a composition of two Manin involutions, and the geometry of the base points was shown to be characteristic for Kahan maps. A similar geometric characterisation for the Kahan discretisation of planar quadratic
Hamiltonian vector fields with a linear Poisson tensor and with a quadratic Hamilton function was given in \cite{PS}.

We conclude with an example from the literature, \cite[section 3]{VGR}, to illustrate how the singularities of the pencil determine the QRT form of the mapping. Using projective coordinates $u=x/z,v=y/z$, the map \cite[equation (8)]{VGR} reads
\[
\rho\left(u,v\right)=\left(
\frac {u(v+1)({q}^{2}-1)+2\,v}{uv(q-1)-u(q+1)+2\,v},-
\frac {uv(q+1)-u(q-1)-2\,v}{uv(q-1)-u(q+1)+2\,v}\right).
\]
It has an invariant of degree 4,
\[
K={\frac { \left( v+1 \right)  \left( uv+u-2\,v \right)  \left( 2\,u-v-1
 \right) }{ \left( v-1 \right) ^{2} \left(
 q^2u(v+1)+2u(u-v-1)+2\,v
 \right) }},
\]
which has two singular base points, namely at $p=(1,1)$ and at $(\infty,0)$. Geometrically the map is understood as the root of a generalised Manin transformation,
\[
\rho = \sigma \circ \iota_1 = \iota_p \circ \sigma,
\]
where $\rho^2=\iota_p\circ\iota_1$. The horizontal switch takes the simple form
\[
\iota_1(u,v)=(v/u,v),
\]
and the symmetry switch $\sigma=\iota_{q+1,1}$ is the projective collineation
\[
\sigma(u,v)=\left(u,v\right)-\left(1+\frac{2q}{2u+(q-1)v-(q+1)}\right)(u-q-1,v-1).
\]
In coordinates
\[
x = -\frac{qv-q+2u+v-3}{2q(v-1)}, y = \frac{q+1-v}{q(v-1)}
\]
the points $(1,1)$ and $(q+1,1)$ are mapped to $(0,\infty)$ and $(\infty,-\infty)$ respectively. The map $\sigma$ becomes the standard symmetry switch, and the integral $K$ is symmetric in $x,y$. Hence, in these coordinates the map is a symmetric QRT map.

\section{Conclusions}
Noting that both Manin transformations and QRT maps are compositions of involutions that switch the 2 points in the intersection of a curve of the invariant pencil with a straight line through a given point, we have constructed classes of such maps which preserve pencils of degree $N=2,3,4$. We have shown how these maps are projectively equivalent to QRT maps, and we have identified classes of maps which are equivalent to roots of QRT maps. For a special configuration of the base points of a cubic pencil, Manin transformations have been shown to arise as the Kahan discretisation of a quadratic planar Hamiltonian vector field in \cite{PSS,PS}. In \cite{vdK}, the current construction is generalised by allowing involutions of the type $\iota_p$, where $p$ is not fixed but lies on a special curve parametrised by the parameter of the pencil, cf. \cite{CKZ} where Manin involutions of this kind were obtained from an open boundary reduction from the $Q1_{\delta=0}$ lattice equation.

\section*{Acknowledgments} This research was supported by the Australian Research Council [DP140100383]. We have extensively used the software Maple \cite{MAP}. This includes drawing the figures, and verifying the densities for the Manin involutions in the cases $N=2,3,4$.

\section*{Appendix A}
We provide the proof of Theorem \ref{T1}.

It is convenient to use abbreviated notation $F_a^{(i)}:=F_a^{(z,\overset{i}{\ldots},z)}(0)$. We start with the Taylor expansion about $z=0$, equation (\ref{TF}), and Taylor expand it about $z=1$:
\[
F_a(z)=\sum_{i=0}^N c_i(z-1)^i,\quad \text{ with } \quad c_i=\sum_{j=i}^N \frac{F_a^{(j)}}{i!(j-i)!}.
\]
As $c_i=0$ for $i<N-2$ we have
\begin{align*}
F_a(z)=\frac{(z-1)^{N-2}}{N!}\Big(&
N(N-1)(F_a^{(N-2)}+F_a^{(N-1)}+\frac12F_a^{(N)})+N(F_a^{(N-1)}\\
&+F_a^{(N)})(z-1)+F_a^{(N)}(z-1)^2\Big).
\end{align*}
Due to $\sum_{i=0}^{N-3}(-1)^ic_i=0$ we have
\[
\frac12(N-1)(N-2)F_a^{(N)}+N(N-2)F_a^{(N-1)}+N(N-1)F_a^{(N-2)}=(-1)^N N!F_a^{(0)}
\]
and hence
\begin{equation} \label{faz}
F_a(z)=\frac{(z-1)^{N-2}}{N!}\left(F_a^{(N)} z^2+(NF_a^{(N-1)}+(N-2)F_a^{(N)})z +(-1)^N N! F_a^{(0)}
\right),
\end{equation}
and similarly for $F_b(z)$. Substituting these into the equation $F_a(z)F_b(0)=F_b(z)F_a(0)$, after dividing out $z(z-1)^{N-2}$ the constant term vanishes, and we are left with a linear equation
\[
(F_a^{(N)} (z+N-2) +NF_a^{(N-1)})F_b^{(0)}=
(F_b^{(N)} (z+N-2) +NF_b^{(N-1)})F_a^{(0)},
\]
which provides
\begin{equation} \label{fir}
z=2-N\left(1+\frac{F_a(0)F_b^{(z,\overset{N-1}{\ldots},z)}(0)-F_a^{(z,
\overset{N-1}{\ldots},z)}(0)F_b(0)}{F_a(0)F_b^{(z,\overset{N}{\ldots},z
)}(0)-F_a^{(z,\overset{N}{\ldots},z)}(0)F_b(0)}\right).
\end{equation}

To get the expression (\ref{sec}) we solve the system $c_i=0$, $0\leq i \leq N-3$. This can be done as follows.
Define $x_{0,j}=(N-j)!c_j$ and  $x_{i+1,j}=\frac{x_{i,j}-x_{i,j+1}}{i+1}$. Explicitly we have, for $0\leq i \leq N-3$,
\[
x_{i,0} = \sum_{j=0}^{N-i} \frac{\prod_{k=0}^{N-i-j-1} (N-i-k)(i+k+1)}{(N-i-j)!}F_a^{(j)},
\]
and the linear combination
\begin{align}
\sum_{h=3}^{k} (-1)^{h+k}\frac{(N-h)!}{(N-k)!}&\binom{k}{h} x_{N-h,0} = F_a^{(k)}-(-1)^k\frac{k!}{2}\left(\binom{N-2}{k-2}F_a^{(2)} \right.\notag \\
&\left.+2\binom{N-1}{k-1}(k-2)F_a^{(1)}+\binom Nk (k-1)(k-2)F_a^{(0)}\right), \label{fak}
\end{align}
and similar for $F_b^{(\bullet)}$. In terms of
\begin{align}
G_n&=F_a^{(0)}F_b^{(n)}-F_a^{(n)}F_b^{(0)} \label{G} \\
&=(-1)^k\frac{k!}{2}\left(\binom{N-2}{k-2}G_2+2\binom{N-1}{k-1}(k-2)G_1\right) \notag
\end{align}
one can show that
\[
\frac{G_{N-1}}{G_N}+1=\frac{2}{N}\frac{(2N-3)G_1+G_2}{(2N-4)G_1+G_2}.
\]
\section*{Appendix B}
We provide a condition that is equivalent to the generalised Manin involution $\iota_p$ given by (\ref{Man}) being anti measure-preserving with density
\[
\rho=\frac{L^{N-3}}{F_a},
\]
where $L=0$ is a line through $p$.

It can be verified that the Jacobian determinant of the map $\iota_p$ equals
\[
\text{Jac}(\iota_p)=\frac { \left( 2 \left( N-1 \right) G_1  +G_2 \right) X }{
 \left( 2 \left( N-2 \right) G_1 +G_2 \right) ^{3}},
 \]
with
\[
X=2
 \left(  \left( c-u \right) G^{(u)}_2
 + \left( d-v \right)G^{(v)}_2 \right) G_1 +4\, \left( N-1
 \right) \left( \left( N-2 \right) G_1 +G_2 \right) G_1 -{G_2}^{2}.
\]
On the other hand, by substituting the expressions for $F^{(N)}$ and $F^{(N-1)}$ as given by (\ref{fak}) into (\ref{faz}) with $z$ given by (\ref{fir}) we find
\[
-\frac{\rho(u,v)}{\rho(\iota_p(u,v))}=\frac { \left( 2 \left( N-1 \right) G_1 +G_2 \right)  Y}{
 \left( 2 \left( N-2 \right) G_1 +G_2 \right) ^{3}},
\]
with
\[
Y=
2\, G_1\left( \frac {F_a^{(1)}}{F_a}\,G_2-\frac{F_a^{(2)}}{F_a}\,G_1 \right)-2 \left( N-2 \right) G_1
 \left( \left( N-1 \right) G_1 +G_2 \right) -{G_2}^{2}
.
\]
We have $Y=X$ if
\begin{equation}\label{CfD}
\begin{split}
\left( c-u \right)G^{(u)}_2 &+ \left( d-v \right)G^{(v)}_2 + F_a^{(1)}F_b^{(2)}-F_a^{(2)}F_b^{(1)}\\
&= 2 \left( N-1 \right)  \left(\left( N-2 \right) G_1 +G_2 \right) +
 \left( N-2 \right)  \left(\left( N-1 \right) G_1 +G_2 \right).
\end{split}
\end{equation}
It is easy, using Maple \cite{MAP}, to verify that condition (\ref{CfD}) is satisfied for pencils of degree
$N=2,3,4$.

\section*{Appendix C}
We prove that no new generalised Manin transformations of the form (\ref{tpq}) are obtained from pencils of degree $N>4$.

\begin{theorem}
Higher degree $N>4$ curves with two distinct points of multiplicity $N-2$ are products of the form $C=L^{N-4}Q$, where $L$ is the line through the two points, and $Q$ a quartic.
\end{theorem}

\begin{proof}
Consider a degree $N=5$ curve $C$ with two distinct points of multiplicity 3. Let $L$ be the line through these points. Near each triple point there is a line which intersects the curve in at least three points, see Figures \ref{F3} and \ref{F4}. Note that while we have drawn the generic case where 3 tangents intersect at each triple point, the statement is still true when some of these tangents are imaginary, e.g. when the curve contains a cusp. If $C$ does not contain $L$ there is a line close to $L$ which intersects $C$ in 6 points, which contradicts $N=5$.

\begin{figure}[h]
\begin{center}
\includegraphics[width=7cm]{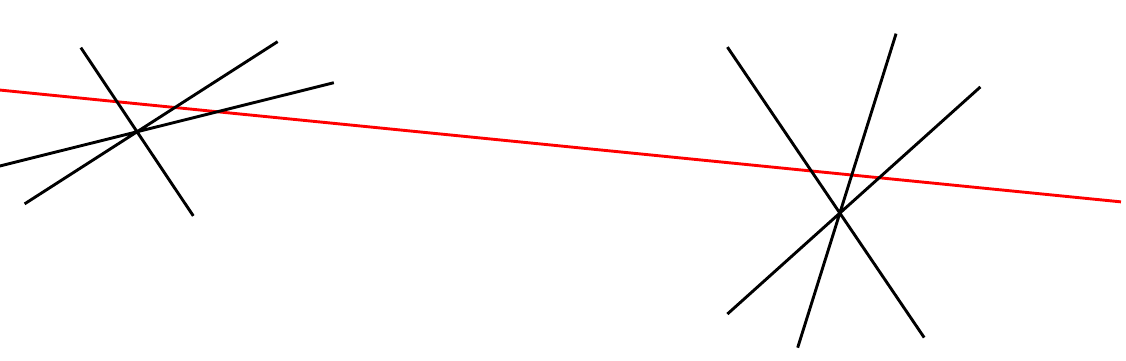}
\end{center}
\caption{\label{F3} A degree 5 curve does not intersect a line in 6 points.}
\end{figure}

\begin{figure}[h]
\begin{center}
\includegraphics[width=7cm]{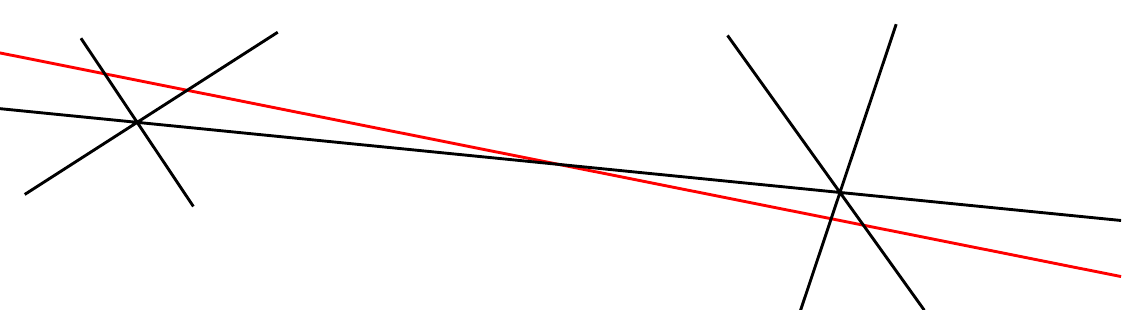}
\end{center}
\caption{\label{F4} Any degree 5 curve with two triple points contains the line through the triple points.}
\end{figure}

Next, let $m$ be the multiplicity of $L$ in a degree $N$ curve $C$.
We need $2(N-2)-m=N$, which implies $m=N-4$.
\end{proof}

This implies that we obtain the same involutions as in the case $N=4$.

\begin{corollary} \label{10}
For higher degree $N>4$ curves with two distinct points of multiplicity $N-2$, the value of $z$, as given by equation (\ref{sec}), does not depend on $N$, for $N\geq4$.
\end{corollary}
\begin{proof}
Consider the degree $N+1$ pencil $\alpha \hat{F}_a(u,v) + \beta \hat{F}_b(u,v)=0$ where $\hat{F}_a(u,v)=F_a(u,v)L(u,v)$, where $F_a$ has degree $N$ and two singular points of multiplicity $N-2$, $(c,d)$ and $(e,f)$, and $L(u,v)=(d-f)(u-e)-(c-e)(v-f)$. We evaluate the functions in (\ref{sec}) at $u+(\tilde{c}-u)z,v+(\tilde{d}-v)z)$, we let $^\prime$ denote differentiation with respect to $z$ and we evaluate at $z=0$. We have
$\hat{F}_a^\prime = F_a^\prime L + L^\prime F_a$ and
$\hat{F}_a^{\prime\prime} = F_a^{\prime\prime} L + 2 L^\prime F_a^\prime$, as $L^{\prime\prime}=0$. Let
\[
K=\dfrac{F_aF_b^{\prime\prime}-F_a^{\prime\prime}F_b}
{F_aF_b^{\prime}-F_a^{\prime}F_b}.
\]
Then
\[
\hat{K}=\frac{\hat{F}_a\hat{F}_b^{\prime\prime}
-\hat{F}_a^{\prime\prime}\hat{F}_b}{\hat{F}_a\hat{F}_b^{\prime}
-\hat{F}_a^{\prime}\hat{F}_b}=K+2\frac{L^\prime}{L},
\]
and
\[
\frac{L^\prime}{L}=-1+\dfrac{\tilde{c}(d-f)+\tilde{d}(e-c)+cf-de}{L}=-1
\]
when $(\tilde{c},\tilde{d})$ equals $(c,d)$ or $(e,f)$. Therefore, from (\ref{sec}),
\[
z_{N+1}=2\left(2(2-N-1)-\hat{K}\right)^{-1}
=2\left(2(2-N)-2-(K-2)\right)^{-1}
=z_N.
\]
\end{proof}

\section*{Appendix D}
Here we give the constants that appear in our formula for quartic polynomials with two double points $(c,d)$ and $(e,f)$, (\ref{quart}):
{\small
\begin{align*}
P&=\left( 4\,{d}^{3}-4\,{f}^{3} \right) a_{{1}}+ \left(3\,
{d}^{3}e-3\,c{f}^{3} \right) a_{{2}}+ \left( 3\,{d}^{3}f-3\,d{f}^{3} \right) a_{{3
}}+ \left(2\,{e}^{2}{d}^{3} -2\,{c}^{2}{f}^{3} \right) a_{{4}} \\
&+
 \left(2\,{d}^{3}ef -2\,cd{f}^{3} \right) a_{{5}}+ \left( 2\,{d}^{3}{
f}^{2}-2\,{d}^{2}{f}^{3} \right) a_{{6}}+ \left( {d}^
{3}{e}^{2}f -{c}^{2}d{f}^{3}\right) a_{{8}}+ \left( {d}^{3}e{f}^{2}-c{d}^{2}{f}^{3}
 \right) a_{{9}}\\[1mm]
Q&=\left(4\,{e}^{3}-4\,{c}^{3} \right) a_{{1}}+ \left(3
\,c{e}^{3}-3\,{c}^{3}e \right) a_{{2}}+ \left( 3\,d{e}^{3}-3\,{c}^{3}f \right) a_
{{3}}+ \left(2\,{c}^{2}{e}^{3}-2\,{c}^{3}{e}^{2} \right) a_{{4}} \\
&+
 \left(2\,cd{e}^{3} -2\,{c}^{3}ef\right) a_{{5}}+ \left(2\,{e}^{3}{d}^{2}-2\,{c}^{3}
{f}^{2} \right) a_{{6}}+ \left({c}
^{2}d{e}^{3} -{c}^{3}{e}^{2}f\right) a_{{8}}+ \left( c{d}^{2}{e}^{3}-{c}^{3}e{f}^{2}
 \right) a_{{9}}
\\[1mm]
R&=\left( 9\,{c}^{3}{d}^{4}{f}^{3}-12\,{c}^{3}{d}^{3}{f}^{4}+3\,{c}^{3}{
f}^{7}+18\,{c}^{2}{d}^{5}e{f}^{2}-24\,{c}^{2}{d}^{4}e{f}^{3}+6\,{c}^{2
}de{f}^{6}-6\,c{d}^{6}{e}^{2}f \right.\\
&\left. +24\,c{d}^{3}{e}^{2}{f}^{4}-18\,c{d}^{2}
{e}^{2}{f}^{5}-3\,{d}^{7}{e}^{3}+12\,{d}^{4}{e}^{3}{f}^{3}-9\,{d}^{3}{
e}^{3}{f}^{4} \right) a_{{1}}+ \left( -2\,{c}^{4}{d}^{3}{f}^{4}+2\,{c}
^{4}{f}^{7}\right.\\
&\left.+5\,{c}^{3}{d}^{4}e{f}^{3}-9\,{c}^{3}{d}^{3}e{f}^{4} +4\,{c}
^{3}de{f}^{6}+12\,{c}^{2}{d}^{5}{e}^{2}{f}^{2}-18\,{c}^{2}{d}^{4}{e}^{
2}{f}^{3}+18\,{c}^{2}{d}^{3}{e}^{2}{f}^{4}\right.\\
&\left. -12\,{c}^{2}{d}^{2}{e}^{2}{f
}^{5}-4\,c{d}^{6}{e}^{3}f+9\,c{d}^{4}{e}^{3}{f}^{3}-5\,c{d}^{3}{e}^{3}
{f}^{4}-2\,{d}^{7}{e}^{4}+2\,{d}^{4}{e}^{4}{f}^{3} \right) a_{{2}}+
 \left( 6\,{c}^{3}{d}^{4}{f}^{4}\right.
\\
&\left.-9\,{c}^{3}{d}^{3}{f}^{5}+3\,{c}^{3}d{
f}^{7}+12\,{c}^{2}{d}^{5}e{f}^{3}-18\,{c}^{2}{d}^{4}e{f}^{4} +6\,{c}^{2
}{d}^{2}e{f}^{6}-6\,c{d}^{6}{e}^{2}{f}^{2}+18\,c{d}^{4}{e}^{2}{f}^{4}\right.
\\
&\left.-
12\,c{d}^{3}{e}^{2}{f}^{5}-3\,{d}^{7}{e}^{3}f+9\,{d}^{5}{e}^{3}{f}^{3}
-6\,{d}^{4}{e}^{3}{f}^{4} \right) a_{{3}} +  \left( {c}^{5}{f}^{7}-4\,{c}^{4}{d}^{3}e{f}^{4}+2\,{c}^{4}de{f}^{6}\right.\\
&\left.+{c}^{3}{d}^{4}{e}^{2}{f}^{3}+6\,{c}^{3}{d}^{3}{e}^{2}{f}^{4}
-6\,{c}^{3}{d}^{2}{e}^{2}{f}^{5}+6\,{c}
^{2}{d}^{5}{e}^{3}{f}^{2}-6\,{c}^{2}{d}^{4}{e}^{3}{f}^{3}-{c}^{2}{d}^{
3}{e}^{3}{f}^{4}-2\,c{d}^{6}{e}^{4}f\right.\\
&\left. +4\,c{d}^{4}{e}^{4}{f}^{3}-{d}^{7}
{e}^{5} \right) a_{{4}}+ \left( -2\,{c}^{4}{d}^{3}{f}^{5}+2\,{c}^{4}d{
f}^{7}+2\,{c}^{3}{d}^{4}e{f}^{4}-6\,{c}^{3}{d}^{3}e{f}^{5}+4\,{c}^{3}{
d}^{2}e{f}^{6}\right.\\
&\left.+6\,{c}^{2}{d}^{5}{e}^{2}{f}^{3}-6\,{c}^{2}{d}^{3}{e}^{2
}{f}^{5}-4\,c{d}^{6}{e}^{3}{f}^{2}+6\,c{d}^{5}{e}^{3}{f}^{3}-2\,c{d}^{
4}{e}^{3}{f}^{4}-2\,{d}^{7}{e}^{4}f +2\,{d}^{5}{e}^{4}{f}^{3} \right) a_{{5}} \\
&+ \left( 3\,{c}^{3}{d}^{4}{f}^{5}-6\,{c}^{3}{d}^{3}{f}^{6}+3\,{c
}^{3}{d}^{2}{f}^{7}+6\,{c}^{2}{d}^{5}e{f}^{4}-12\,{c}^{2}{d}^{4}e{f}^{
5}+6\,{c}^{2}{d}^{3}e{f}^{6}-6\,c{d}^{6}{e}^{2}{f}^{3}\right.
\\
&\left. +12\,c{d}^{5}{e}
^{2}{f}^{4}-6\,c{d}^{4}{e}^{2}{f}^{5}-3\,{d}^{7}{e}^{3}{f}^{2}+6\,{d}^
{6}{e}^{3}{f}^{3}-3\,{d}^{5}{e}^{3}{f}^{4} \right) a_{{6}}+ \left( {c}
^{5}d{f}^{7}-4\,{c}^{4}{d}^{3}e{f}^{5}\right.\\
&\left. +2\,{c}^{4}{d}^{2}e{f}^{6}+4\,{c
}^{3}{d}^{4}{e}^{2}{f}^{4}-3\,{c}^{3}{d}^{3}{e}^{2}{f}^{5}+3\,{c}^{2}{
d}^{5}{e}^{3}{f}^{3}-4\,{c}^{2}{d}^{4}{e}^{3}{f}^{4}-2\,c{d}^{6}{e}^{4
}{f}^{2}+4\,c{d}^{5}{e}^{4}{f}^{3}\right.\\
&\left. -{d}^{7}{e}^{5}f \right) a_{{8}}+
 \left( -2\,{c}^{4}{d}^{3}{f}^{6}+2\,{c}^{4}{d}^{2}{f}^{7}-{c}^{3}{d}^
{4}e{f}^{5}+{c}^{3}{d}^{3}e{f}^{6}+6\,{c}^{2}{d}^{5}{e}^{2}{f}^{4}-6\,
{c}^{2}{d}^{4}{e}^{2}{f}^{5}\right.\\
&\left. -c{d}^{6}{e}^{3}{f}^{3}+c{d}^{5}{e}^{3}{f}
^{4}-2\,{d}^{7}{e}^{4}{f}^{2}+2\,{d}^{6}{e}^{4}{f}^{3} \right) a_{{9}}
+ \left( {c}^{5}{d}^{2}{f}^{7}-2\,{c}^{4}{d}^{3}e{f}^{6}+{c}^{3}{d}^{4
}{e}^{2}{f}^{5}\right.\\
&\left. -{c}^{2}{d}^{5}{e}^{3}{f}^{4}+2\,c{d}^{6}{e}^{4}{f}^{3}
-{d}^{7}{e}^{5}{f}^{2} \right) a_{{13}}
\end{align*}

\begin{align*}
S&=\left( 12\,{c}^{4}{d}^{2}{f}^{4}-12\,{c}^{4}{d}^{3}{f}^{3}-24\,{c}^{
3}{d}^{4}e{f}^{2}+24\,{c}^{3}{d}^{3}e{f}^{3}+36\,{c}^{2}{d}^{4}{e}^{2}
{f}^{2}-36\,{c}^{2}{d}^{2}{e}^{2}{f}^{4} \right. \\
&\left.-24\,c{d}^{3}{e}^{3}{f}^{3}+24
\,c{d}^{2}{e}^{3}{f}^{4}-12\,{d}^{4}{e}^{4}{f}^{2}+12\,{d}^{3}{e}^{4}{
f}^{3} \right) a_{{1}}+ \left( 3\,{c}^{5}{d}^{2}{f}^{4}-6\,{c}^{4}{d}^
{3}e{f}^{3}\right. \\
&\left.+9\,{c}^{4}e{f}^{4}{d}^{2}-15\,{c}^{3}{d}^{4}{e}^{2}{f}^{2}
+18\,{c}^{3}{d}^{3}{e}^{2}{f}^{3}-27\,{c}^{3}{d}^{2}{e}^{2}{f}^{4}+27
\,{c}^{2}{d}^{4}{e}^{3}{f}^{2}-18\,{c}^{2}{d}^{3}{e}^{3}{f}^{3}\right. \\
&\left.+15\,{c
}^{2}{d}^{2}{e}^{3}{f}^{4}-9\,c{e}^{4}{f}^{2}{d}^{4}+6\,c{d}^{3}{e}^{4
}{f}^{3}-3\,{d}^{4}{e}^{5}{f}^{2} \right) a_{{2}}+ \left( -8\,{c}^{4}{
d}^{3}{f}^{4}+9\,{c}^{4}{d}^{2}{f}^{5}-{c}^{4}{f}^{7}\right. \\
&\left.-16\,{c}^{3}{d}^{
4}e{f}^{3}+18\,{c}^{3}{d}^{3}e{f}^{4}-2\,{c}^{3}de{f}^{6}+3\,{c}^{2}{d
}^{5}{e}^{2}{f}^{2}+27\,{c}^{2}{d}^{4}{e}^{2}{f}^{3}-27\,{c}^{2}{d}^{3
}{e}^{2}{f}^{4}\right. \\
&\left.-3\,{c}^{2}{d}^{2}{e}^{2}{f}^{5}+2\,c{d}^{6}{e}^{3}f-18
\,c{d}^{4}{e}^{3}{f}^{3}+16\,c{d}^{3}{e}^{3}{f}^{4}+{d}^{7}{e}^{4}-9\,
{d}^{5}{e}^{4}{f}^{2}+8\,{d}^{4}{e}^{4}{f}^{3} \right) a_{{3}}\\
&+
 \left( 6\,{c}^{5}{d}^{2}e{f}^{4}-12\,{c}^{4}{d}^{2}{e}^{2}{f}^{4}-6\,
{c}^{3}{d}^{4}{e}^{3}{f}^{2}+6\,{c}^{3}{d}^{2}{e}^{3}{f}^{4}+12\,{c}^{
2}{d}^{4}{e}^{4}{f}^{2}-6\,c{d}^{4}{e}^{5}{f}^{2} \right) a_{{4}}\\
&+
 \left( 3\,{c}^{5}{d}^{2}{f}^{5}-{c}^{5}{f}^{7}-2\,{c}^{4}{d}^{3}e{f}^
{4}+6\,{c}^{4}{d}^{2}e{f}^{5}-2\,{c}^{4}de{f}^{6}-7\,{c}^{3}{d}^{4}{e}
^{2}{f}^{3}-6\,{c}^{3}{d}^{3}{e}^{2}{f}^{4}\right. \\
&\left.-3\,{c}^{3}{d}^{2}{e}^{2}{f
}^{5}+3\,{c}^{2}{d}^{5}{e}^{3}{f}^{2}+6\,{c}^{2}{d}^{4}{e}^{3}{f}^{3}+
7\,{c}^{2}{d}^{3}{e}^{3}{f}^{4}+2\,c{d}^{6}{e}^{4}f-6\,c{d}^{5}{e}^{4}
{f}^{2}+{d}^{7}{e}^{5}\right.\\
&\left.+2\,c{d}^{4}{e}^{4}{f}^{3}-3\,{d}^{5}{e}^{5}{f}^
{2} \right) a_{{5}}+ \left( -4\,{c}^{4}{d}^{3}{f}^{5}+6\,{c}^{4}{d}^{2
}{f}^{6}-2\,{c}^{4}d{f}^{7}-8\,{c}^{3}{d}^{4}e{f}^{4}+12\,{c}^{3}{d}^{
3}e{f}^{5}\right. \\
&\left.-4\,{c}^{3}{d}^{2}e{f}^{6}+6\,{c}^{2}{d}^{5}{e}^{2}{f}^{3}-6
\,{c}^{2}{d}^{3}{e}^{2}{f}^{5}+4\,c{d}^{6}{e}^{3}{f}^{2}-12\,c{d}^{5}{
e}^{3}{f}^{3}+8\,c{d}^{4}{e}^{3}{f}^{4}+2\,{d}^{7}{e}^{4}f\right. \\
& \left.-6\,{d}^{6}{
e}^{4}{f}^{2}+4\,{d}^{5}{e}^{4}{f}^{3} \right) a_{{6}}+ \left( -{c}^{6
}{f}^{7}+6\,{c}^{5}{d}^{2}e{f}^{5}-2\,{c}^{5}de{f}^{6}-5\,{c}^{4}{d}^{
3}{e}^{2}{f}^{4}-4\,{c}^{3}{d}^{4}{e}^{3}{f}^{3}\right. \\
&\left.+4\,{c}^{3}{d}^{3}{e}^
{3}{f}^{4}+5\,{c}^{2}{d}^{4}{e}^{4}{f}^{3}+2\,c{d}^{6}{e}^{5}f-6\,c{d}
^{5}{e}^{5}{f}^{2}+{d}^{7}{e}^{6} \right) a_{{8}}+ \left( 3\,{c}^{5}{d
}^{2}{f}^{6}-2\,{c}^{5}d{f}^{7}\right. \\
&\left.+2\,{c}^{4}{d}^{3}e{f}^{5}-{c}^{4}{d}^{
2}e{f}^{6}-8\,{c}^{3}{d}^{4}{e}^{2}{f}^{4}+8\,{c}^{2}{d}^{4}{e}^{3}{f}
^{4}+c{d}^{6}{e}^{4}{f}^{2}-2\,c{d}^{5}{e}^{4}{f}^{3}+2\,{d}^{7}{e}^{5
}f\right. \\
&\left.-3\,{d}^{6}{e}^{5}{f}^{2} \right) a_{{9}}+ \left( -2\,{c}^{6}d{f}^{7
}+2\,{c}^{5}{d}^{2}e{f}^{6}+2\,{c}^{4}{d}^{3}{e}^{2}{f}^{5}-2\,{c}^{2}
{d}^{5}{e}^{4}{f}^{3}-2\,c{d}^{6}{e}^{5}{f}^{2}\right. \\
&\left.+2\,{d}^{7}{e}^{6}f
 \right) a_{{13}}
\\[1mm]
T&=\left( 24\,{c}^{4}{d}^{3}{e}^{2}f-36\,{c}^{4}{d}^{2}{e}^{2}{f}^{2}+12
\,{c}^{4}{e}^{2}{f}^{4}+12\,{c}^{3}{d}^{4}{e}^{3}-24\,{c}^{3}{d}^{3}{e
}^{3}f+24\,{c}^{3}d{e}^{3}{f}^{3}\right. \\
&\left.-12\,{c}^{3}{e}^{3}{f}^{4}-12\,{c}^{2
}{d}^{4}{e}^{4}+36\,{c}^{2}{d}^{2}{e}^{4}{f}^{2}-24\,{c}^{2}d{e}^{4}{f
}^{3} \right) a_{{1}}+ \left( 9\,{c}^{5}{e}^{2}{f}^{4}-{c}^{7}{f}^{4}-2\,{c}^{6}de{f}^{3}\right. \\
&  \left.-3\,{
c}^{5}{d}^{2}{e}^{2}{f}^{2}+16\,{c}^{4}{d}^{3
}{e}^{3}f-27\,{c}^{4}{d}^{2}{e}^{3}{f}^{2}+18\,{c}^{4}d{e}^{3}{f}^{3}-
8\,{c}^{4}{e}^{3}{f}^{4}+8\,{c}^{3}{d}^{4}{e}^{4}+{d}^{4}{e}^{7}\right. \\
&  \left.
-18\,{c}^{3}{d}^{3}{e
}^{4}f+27\,{c}^{3}{d}^{2}{e}^{4}{f}^{2}-16\,{c}^{3}d{e}^{4}{f}^{3}-9\,
{c}^{2}{d}^{4}{e}^{5}+3\,{c}^{2}{d}^{2}{e}^{5}{f}^{2}+2\,c{d}^{3}{e}^{
6}f \right) a_{{2}} \\
&+ \left( 15\,{c}^{4}{d}^{3}{e}^{2}{f
}^{2}-27\,{c}^{4}{d}^{2}{e}^{2}{f}^{3}+9\,{c}^{4}d{e}^{2}{f}^{4}+3\,{c
}^{4}{e}^{2}{f}^{5}+6\,{c}^{3}{d}^{4}{e}^{3}f-18\,{c}^{3}{d}^{3}{e}^{3
}{f}^{2}\right. \\
&  \left.+18\,{c}^{3}{d}^{2}{e}^{3}{f}^{3}-6\,{c}^{3}d{e}^{3}{f}^{4}-3
\,{c}^{2}{d}^{5}{e}^{4}-9\,{c}^{2}{d}^{4}{e}^{4}f+27\,{c}^{2}{d}^{3}{e
}^{4}{f}^{2}-15\,{c}^{2}{d}^{2}{e}^{4}{f}^{3} \right) a_{{3}}\\
&+ \left(
-2\,{c}^{7}e{f}^{4}-4\,{c}^{6}d{e}^{2}{f}^{3}+6\,{c}^{6}{e}^{2}{f}^{4}
-6\,{c}^{5}{d}^{2}{e}^{3}{f}^{2}+12\,{c}^{5}d{e}^{3}{f}^{3}-4\,{c}^{5}
{e}^{3}{f}^{4}+8\,{c}^{4}{d}^{3}{e}^{4}f\right. \\
&  \left.-8\,{c}^{4}d{e}^{4}{f}^{3}+4\,
{c}^{3}{d}^{4}{e}^{5}-12\,{c}^{3}{d}^{3}{e}^{5}f+6\,{c}^{3}{d}^{2}{e}^
{5}{f}^{2}-6\,{c}^{2}{d}^{4}{e}^{6}+4\,{c}^{2}{d}^{3}{e}^{6}f+2\,c{d}^
{4}{e}^{7} \right) a_{{4}}\\
&+ \left( -{c}^{7}{f}^{5}-2\,{c}^{6}de{f}^{4}
-3\,{c}^{5}{d}^{2}{e}^{2}{f}^{3}+6\,{c}^{5}d{e}^{2}{f}^{4}+3\,{c}^{5}{
e}^{2}{f}^{5}+7\,{c}^{4}{d}^{3}{e}^{3}{f}^{2}-6\,{c}^{4}{d}^{2}{e}^{3}
{f}^{3}\right. \\
&  \left.-2\,{c}^{4}d{e}^{3}{f}^{4}+2\,{c}^{3}{d}^{4}{e}^{4}f+6\,{c}^{3}
{d}^{3}{e}^{4}{f}^{2}-7\,{c}^{3}{d}^{2}{e}^{4}{f}^{3}-3\,{c}^{2}{d}^{5
}{e}^{5}-6\,{c}^{2}{d}^{4}{e}^{5}f+3\,{c}^{2}{d}^{3}{e}^{5}{f}^{2}\right.\\
&  \left.+2\,
c{d}^{4}{e}^{6}f+{d}^{5}{e}^{7} \right) a_{{5}}+ \left( 6\,{c}^{4}{d}^
{3}{e}^{2}{f}^{3}-12\,{c}^{4}{d}^{2}{e}^{2}{f}^{4}+6\,{c}^{4}d{e}^{2}{
f}^{5}-6\,{c}^{2}{d}^{5}{e}^{4}f-6\,{
c}^{2}{d}^{3}{e}^{4}{f}^{3}\right.
\\
&  \left.+12\,{c}^{2}{d}^{4}{e}^{4}{f}^{2} \right) a_{{6}}+ \left( -2\,{c}^{7}e{f}^{5
}-{c}^{6}d{e}^{2}{f}^{4}+3\,{c}^{6}{e}^{2}{f}^{5}+2\,{c}^{5}d{e}^{3}{f
}^{4}+8\,{c}^{4}{d}^{3}{e}^{4}{f}^{2}+2\,c{d}^{5}{e}^{7}\right. \\
&  \left.-8\,{c}^{4}{d}^{2}{e}^{4}{f}^{3}
-2\,{c}^{3}{d}^{4}{e}^{5}f-3\,{c}^{2}{d}^{5}{e}^{6}+{c}^{2}{d}^{4}{e}^{
6}f \right) a_{{8}}+ \left( -{c}^{7}{f}^{6}-2\,{c}^
{6}de{f}^{5}+6\,{c}^{5}d{e}^{2}{f}^{5}\right. \\
&  \left.+4\,{c}^{4}{d}^{3}{e}^{3}{f}^{3}
-5\,{c}^{4}{d}^{2}{e}^{3}{f}^{4}+5\,{c}^{3}{d}^{4}{e}^{4}{f}^{2}-4\,{c
}^{3}{d}^{3}{e}^{4}{f}^{3}-6\,{c}^{2}{d}^{5}{e}^{5}f +2\,c{d}^{5}{e}^{6
}f+{d}^{6}{e}^{7} \right) a_{{9}} \\
&+ \left( -2\,{c}^{7}e{f}^{6}+2\,{c}^{
6}d{e}^{2}{f}^{5}+2\,{c}^{5}{d}^{2}{e}^{3}{f}^{4}-2\,{c}^{3}{d}^{4}{e}
^{5}{f}^{2}-2\,{c}^{2}{d}^{5}{e}^{6}f+2\,c{d}^{6}{e}^{7} \right) a_{{
13}}
\\[1mm]
U&=\left( 3\,{c}^{7}{f}^{3}+6\,{c}^{6}de{f}^{2}-18\,{c}^{5}{d}^{2}{e}^{2
}f-9\,{c}^{4}{d}^{3}{e}^{3}+24\,{c}^{4}{d}^{2}{e}^{3}f-12\,{c}^{4}{e}^
{3}{f}^{3}+12\,{c}^{3}{d}^{3}{e}^{4}\right. \\
&  \left.-24\,{c}^{3}d{e}^{4}{f}^{2}+9\,{c}
^{3}{e}^{4}{f}^{3}+18\,{c}^{2}d{e}^{5}{f}^{2}-6\,c{d}^{2}{e}^{6}f-3\,{
d}^{3}{e}^{7} \right) a_{{1}}+ \left( 3\,{c}^{7}e{f}^{3}+6\,{c}^{6}d{e
}^{2}{f}^{2}\right. \\
&  \left.
-12\,{c}^{5}{d}^{2}{e}^{3}f-9\,{c}^{5}{e}^{3}{f}^{3}-6\,{c
}^{4}{d}^{3}{e}^{4}+18\,{c}^{4}{d}^{2}{e}^{4}f-18\,{c}^{4}d{e}^{4}{f}^
{2}+6\,{c}^{4}{e}^{4}{f}^{3}+9\,{c}^{3}{d}^{3}{e}^{5}\right. \\
&  \left.+12\,{c}^{3}d{e}^
{5}{f}^{2}-6\,{c}^{2}{d}^{2}{e}^{6}f-3\,c{d}^{3}{e}^{7} \right) a_{{2}
}+ \left( 2\,{c}^{7}{f}^{4}+4\,{c}^{6}de{f}^{3}-12\,{c}^{5}{d}^{2}{e}^
{2}{f}^{2}-5\,{c}^{4}{d}^{3}{e}^{3}f\right. \\
&  \left.+18\,{c}^{4}{d}^{2}{e}^{3}{f}^{2}-
9\,{c}^{4}d{e}^{3}{f}^{3}-2\,{c}^{4}{e}^{3}{f}^{4}+2\,{c}^{3}{d}^{4}{e
}^{4}+9\,{c}^{3}{d}^{3}{e}^{4}f-18\,{c}^{3}{d}^{2}{e}^{4}{f}^{2}+5\,{c
}^{3}d{e}^{4}{f}^{3}\right. \\
&  \left.+12\,{c}^{2}{d}^{2}{e}^{5}{f}^{2}-4\,c{d}^{3}{e}^{
6}f-2\,{d}^{4}{e}^{7} \right) a_{{3}}+ \left( 3\,{c}^{7}{e}^{2}{f}^{3}
+6\,{c}^{6}d{e}^{3}{f}^{2}-6\,{c}^{6}{e}^{3}{f}^{3}-6\,{c}^{5}{d}^{2}{
e}^{4}f\right.\\
&  \left.-12\,{c}^{5}d{e}^{4}{f}^{2}+3\,{c}^{5}{e}^{4}{f}^{3}-3\,{c}^{4}
{d}^{3}{e}^{5}+12\,{c}^{4}{d}^{2}{e}^{5}f+6\,{c}^{4}d{e}^{5}{f}^{2}+6
\,{c}^{3}{d}^{3}{e}^{6}-6\,{c}^{3}{d}^{2}{e}^{6}f \right. \\
&  \left.-3\,{c}^{2}{d}^{3}{e}^{7} \right) a_{{4}}
+ \left( 2\,{c}^{7}e{f}^{4}+4\,{c}^{6}d{e}^{2}{f}^
{3}-6\,{c}^{5}{d}^{2}{e}^{3}{f}^{2}-6\,{c}^{5}d{e}^{3}{f}^{3}-2\,{c}^{
5}{e}^{3}{f}^{4}-2\,{c}^{4}{d}^{3}{e}^{4}f\right.
\end{align*}

\begin{align*}
& \left.+2\,{c}^{4}d{e}^{4}{f}^{3}+2
\,{c}^{3}{d}^{4}{e}^{5}+6\,{c}^{3}{d}^{3}{e}^{5}f+6\,{c}^{3}{d}^{2}{e}
^{5}{f}^{2}-4\,{c}^{2}{d}^{3}{e}^{6}f-2\,c{d}^{4}{e}^{7} \right) a_{{5
}}+ \left( {c}^{7}{f}^{5}\right. \\
&  \left.
+2\,{c}^{6}de{f}^{4}-6\,{c}^{5}{d}^{2}{e}^{2}
{f}^{3}-{c}^{4}{d}^{3}{e}^{3}{f}^{2}+6\,{c}^{4}{d}^{2}{e}^{3}{f}^{3}-4
\,{c}^{4}d{e}^{3}{f}^{4}+4\,{c}^{3}{d}^{4}{e}^{4}f-6\,{c}^{3}{d}^{3}{e
}^{4}{f}^{2}\right. \\
&  \left.
+{c}^{3}{d}^{2}{e}^{4}{f}^{3}+6\,{c}^{2}{d}^{3}{e}^{5}{f}^
{2}-2\,c{d}^{4}{e}^{6}f-{d}^{5}{e}^{7} \right) a_{{6}} + \left( 2\,{c}^
{7}{e}^{2}{f}^{4}+{c}^{6}d{e}^{3}{f}^{3}-2\,{c}^{6}{e}^{3}{f}^{4}\right. \\
&  \left.-6\,{
c}^{5}{d}^{2}{e}^{4}{f}^{2}-{c}^{5}d{e}^{4}{f}^{3}+{c}^{4}{d}^{3}{e}^{
5}f+6\,{c}^{4}{d}^{2}{e}^{5}{f}^{2}+2\,{c}^{3}{d}^{4}{e}^{6}-{c}^{3}{d
}^{3}{e}^{6}f -2\,{c}^{2}{d}^{4}{e}^{7} \right) a_{{8}} \\
&  + \left( {c}^{7}
e{f}^{5}+2\,{c}^{6}d{e}^{2}{f}^{4}-3\,{c}^{5}{d}^{2}{e}^{3}{f}^{3}-4\,
{c}^{5}d{e}^{3}{f}^{4}-4\,{c}^{4}{d}^{3}{e}^{4}{f}^{2}+4\,{c}^{4}{d}^{
2}{e}^{4}{f}^{3}+4\,{c}^{3}{d}^{4}{e}^{5}f \right. \\
&  \left.+3\,{c}^{3}{d}^{3}{e}^{5}{f}
^{2}-2\,{c}^{2}{d}^{4}{e}^{6}f-c{d}^{5}{e}^{7} \right) a_{{9}}+
 \left( {c}^{7}{e}^{2}{f}^{5}-2\,{c}^{6}d{e}^{3}{f}^{4}+{c}^{5}{d}^{2}
{e}^{4}{f}^{3}-{c}^{4}{d}^{3}{e}^{5}{f}^{2} \right. \\
&  \left.+2\,{c}^{3}{d}^{4}{e}^{6}f -{c}^{2}{d}^{5}{e}^{7} \right.) a_{{13}}
\end{align*}
}

\section*{Appendix E}
We provide a class of pencils which admit a fractional linear symmetry switch and show that each curve of such a pencil is a product of lines.

The projective collineation (\ref{flsw}) is an involution for solutions of
\begin{equation} \label{bils}
\begin{split}
ab+be+ch=0,\quad ac+bf+ci=0,\quad ad+de+fg=0,\quad ag+dh+gi=0\\
bg+eh+hi=0,\quad cd+ef+fi=0,\quad a^2+bd=fh+i^2,\quad e^2+bd=cg+i^2.
\end{split}
\end{equation}
Assuming that $b\neq0$, the highest dimensional family of solutions\footnote{Lower dimensional solutions can be obtained by taking $b=0$ and either $c=0$ or $h=0$.} to (\ref{bils}) can be parameterised in terms of $b,c,e,h,i$
by
\begin{equation} \label{hdf}
a=-e-{\frac {ch}{b}},\quad g=-{\frac {h \left( e+i \right) }{b}},\quad f={\frac {c
 \left( be-bi+ch \right) }{{b}^{2}}},\quad  d=-{\frac { \left( e+i \right)
 \left( be-bi+ch \right) }{{b}^{2}}}.
\end{equation}
We reparametrise\footnote{This reparametrisation is invertible when $bh(bh+ch-bi)(be+ch-bi)\neq0$. In particular, this means that the linear switch from the previous section (where $h=0$) is not included.} the solution (\ref{hdf}) in terms of parameters $\alpha,\beta,\gamma,P,Q$
\begin{align*}
b&=hP,& c&={\frac {hP \left( \beta PQ - (\alpha +\gamma) Q-2\,\alpha \right)}{\alpha+\gamma}},\\
e&={\frac {h \left( \alpha Q+\gamma Q+\alpha \right)}{\alpha+\gamma}},&
i&={\frac { h\left(\beta PQ-(\alpha+\gamma) Q-\alpha \right) }{\alpha+\gamma}}.
\end{align*}
The parameters $P,Q$ play a special role; defining $Y=(P,Q)$ the projective collineation takes the form
\begin{equation} \label{nsw}
\sigma:U\rightarrow U+z(U-Y),\qquad z=-1+\frac{\alpha}{(\alpha+\gamma)v-\beta Qu+\delta},
\end{equation}
with
\begin{equation} \label{D}
\delta=(\beta P -\gamma)Q - \alpha(Q+1).
\end{equation}
The constraint (\ref{D}) ensures that $\sigma$ (\ref{nsw}) is an involution. The form of (\ref{nsw}) shows that $\sigma$ preserves lines throught $Y$. We take
\[
P=\frac{BF-CE}{AE-BD},\qquad Q=\frac{CD-AF}{AE-BD},
\]
so that $Y=(P,Q)$ is the intersection point
of the lines $S=0$ and $T=0$, where
\[
{S}=Au+Bv+C,\qquad {T}=Du+Ev+F.
\]
Having fixed $Y$, the four parameter family of projective collineations (\ref{nsw})
leaves the ratio ${S}/{T}$ invariant, and a three parameter subfamily, defined by (\ref{D}), consists of involutions.

Using ${S}$, ${T}$ we can build pencils of fixed degree which are invariant under $\sigma$ (\ref{nsw}). For $N=2$ we have $P_{\alpha,\beta}(u,v)=0$ where
\begin{equation} \label{qp}
F_a=a_1 {S}^2+ a_2 {S}{T}+ a_3 {T}^2,\qquad
F_b=b_1 {S}^2+ b_2 {S}{T}+ b_3 {T}^2.
\end{equation}
The point $Y$ is a double point of the pencil $P_{\alpha,\beta}(u,v)=0$. Because the degree of the pencil is two, all curves are singular, i.e. each curve factorises into a product of lines. If $\alpha,\beta$ are such that $P_{\alpha,\beta}(\hat{u},\hat{v})=0$, then $ P_{\alpha,\beta}(u,v)=LK$. If $L=0$ is the line through $Y$ and $\hat{U}=(\hat{u},\hat{v})$, then $K=0$ is the line through $Y$ with direction
\[
\begin{pmatrix}
a_1b_2-a_2b_1 & a_2b_3-a_3b_2 & a_1b_3-a_3b_1
\end{pmatrix}
\begin{pmatrix}
-B\hat{S} & A\hat{S} \\
-E\hat{T} & D\hat{T} \\
-(B\hat{T}+E\hat{S})& A\hat{T}+ D\hat{S}
\end{pmatrix}.
\]
Choosing involution points $p$ and $\sigma(p)$, for some $\alpha,\beta,\gamma$, and $\delta$ given by (\ref{D}), the map $\iota_{\sigma(p)}\circ \iota_p$ admits a root.

\begin{theorem}
The root $\rho_p=\sigma \circ \iota_p$, where $\sigma$ is given by (\ref{nsw}), with (\ref{D}), and $\iota_p$ by (\ref{evol}), is an integrable map of the plane. It preserves each curve of the quadratic pencil $P_{\alpha,\beta}(u,v)=0$ with (\ref{pen}) and (\ref{qp}), and it is measure-preserving with density  $(L F_a)^{-1}$ where, with $p=(c,d)$,
\begin{equation} \label{LL}
L=(d-Q)(u-P)-(c-P)(v-Q),
\end{equation}
so that $L=0$ is the line through $p$ and $Y$.
\end{theorem}

For $N=3$ we have that
\[
P_{\alpha,\beta}(u,v)=\alpha(a_1 {S}^3+ a_2 {S}^2{T}+ a_3 {S}{T}^2 + a_4{T}^3)+\beta(b_1 {S}^3+ b_2 {S}^2{T}+ b_3 {S}{T}^2+b_4{T}^3)
\]
admits the symmetry switch (\ref{nsw}). We
require that the involution point $p=(c,d)$ is a base point of the pencil. Because $Y$ is a triple point, each curve is a product of three lines, with common intersection point $Y$. Thus the line $L$ through $p$ and $Y$ (\ref{LL}) is contained in each curve, we have $P_{\alpha,\beta}(u,v)=LZ$, where $Z$ is a quadratic polynomial with a double point at $Y$. No new maps which admit a root are obtained, other than the ones already obtained in the $N=2$ case. Similarly no other maps are obtained
in the $N>3$ case where the requirement of the involution point $p$ being a singular point of multiplicity $N-2$ leads to the factorisation $P_{\alpha,\beta}(u,v)=L^{N-2}Z$.

\begin{example} \label{20}
We take $S=u+12v+2$ and $T=2u-4v-3$, so the lines $S=0$ and $T=0$ intersect in $Y=(1,-1/4)$. Taking $N=2$, $a_i=i+1$, $b_i=4-i$, the point $s=(1,3/7)$ lies on the curve $P_{34,-31}(u,v)=35(u-1)(9u-88v-31)=0$. Choosing $\alpha=-6,\beta=20,\gamma=6$ gives $d=1$ and
\[
\sigma:(u,v)\rightarrow \left(\frac{7-u}{5u+1},-\frac14\frac{5u+24v+7}{5u+1}\right).
\]
One verifies that
\[
P_{34,-31}(\sigma(u,v))=\frac{1260(u-1)(9u-88v-31)}{(5u+1)^2}.
\]
We choose the point $p=(2,1)$ as involution point, and we find $r=\iota_p(s)=-(129,115)/289$. The points
\[
\sigma(p)=(5/11,-41/44),\qquad
\sigma(s)=(1,-13/14),\qquad
\sigma(r)=(-538/89,-691/712)
\]
lie on a straight line, see Figure \ref{FFL}.
It can also be checked that
\[
\sigma(\iota_p(\sigma(r)))=\sigma(1,4325/5728)=(1,-7189/5728)=\iota_{\sigma(p)}(r).
\]

\begin{figure}[htb]
\begin{center}
\includegraphics[width=8cm]{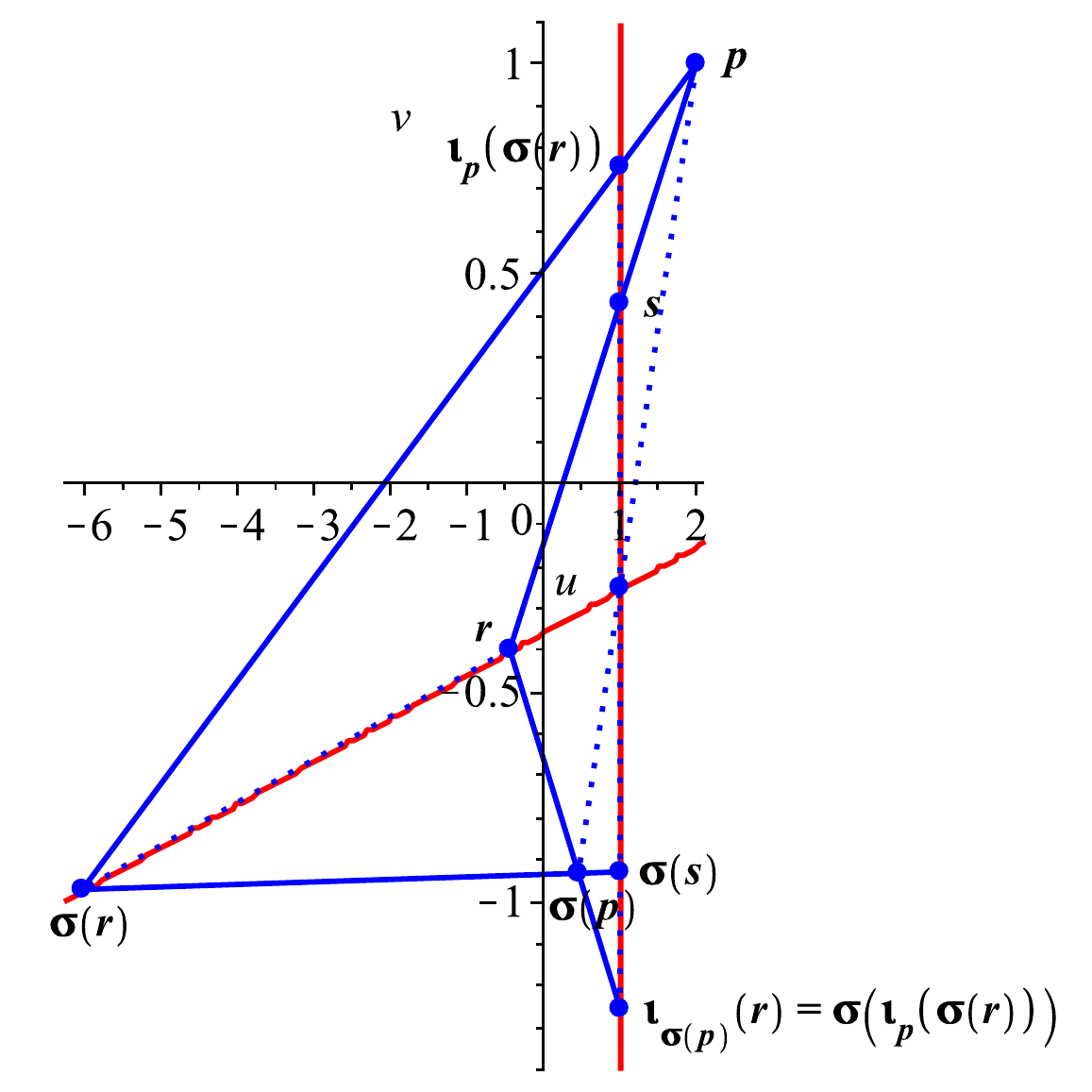}
\end{center}
\caption{\label{FFL} A product of lines admitting fractional linear symmetries, cf. Example \ref{20}.}
\end{figure}

\end{example}


\end{document}